\documentclass[aps,pre,amsmath,twocolumn,superscriptaddress,showpacs]{revtex4-2}
\usepackage{calc}
\usepackage{amsmath}
\usepackage{amsfonts}
\usepackage{amssymb}
\usepackage{graphicx}
\usepackage{hyperref}
\usepackage{soul}
\usepackage{mathtools}
\usepackage{verbatim}
\usepackage{epstopdf}
\usepackage{subfigure}
\usepackage{latexsym}
\usepackage{dcolumn}
\usepackage{epsf}
\usepackage{float}
\usepackage[table]{xcolor}
\usepackage{multirow}
\usepackage[toc]{appendix}

\usepackage{color} 

\usepackage{graphicx}
\usepackage{graphicx,epstopdf}
\usepackage{xcolor}
\usepackage{dcolumn}
\usepackage{bm}
\usepackage{mathrsfs} 
\usepackage{mathrsfs}

\usepackage{colortbl}
\usepackage[table]{xcolor}
\usepackage{amsthm}

\newtheorem{lemma}{Lemma}

\theoremstyle{definition}

\theoremstyle{remark}

\begin{document}

\preprint{APS/123-QED}

\title{Topologically protected edge oscillations in nonlinear dynamical units}

\author{Sayantan Nag Chowdhury}
\email{snagchowdh@constructor.university}
\affiliation{School of Science, Constructor University, P.O.Box 750561, 28725 Bremen, Germany}
\author{Hildegard Meyer-Ortmanns}
\email{hmeyerortm@constructor.university}
\affiliation{School of Science, Constructor University, P.O.Box 750561, 28725 Bremen, Germany} 
\affiliation{Complexity Science Hub, Vienna, Metternichgasse 8, 1030 Wien, Austria}
\date{\today}

\begin{abstract}
Many examples from quantum and classical physics are known where topological protection is responsible for the robustness of the dynamics. Less explored is  the role of topological protection in the context of classical oscillatory systems. As on-site dynamics, we consider prototypical oscillator models with possible applications in biochemical systems. However, our choice of coupling geometry is inspired  by models from condensed matter physics which -in isolation- guarantee  non-trivial topology in momentum space. 
 We choose directed couplings between units on a two-dimensional grid, alternating between weak and strong values, such that oscillations become localized at the edges of the grid while bulk units transition to oscillation-death states, resulting in a frequency chimera-like state. These patterns are resilient to parameter mismatches, additive noise, and structural defects. To explain the robustness of these edge oscillations we use  topological characteristics  and calculate Zak phases. As it turns out, the edge-localized oscillations result from a  bulk-boundary correspondence that applies to our system even though our derived effective Hamiltonian is non-Hermitian.  By appropriately tuning the system parameters, it is possible to control which regions of the two-dimensional grid are in an oscillatory state and which settle to oscillation-death states.
\end{abstract}

\maketitle


\section{\label{sec:intro}Introduction}
Topological protection is a concept in physics. When it applies, the associated phenomena are rather stable. The space with non-trivial topology may be phase space, configuration space, or momentum space. The common mathematical origin are topological invariants. These invariants cannot be changed by continuous deformations of physical parameters of the dynamical system. A prototypical example is the quantum Hall effect, where electrons make cyclotron orbits in the bulk that lead to unidirectional currents along the boundaries \cite{cohen2019geometric}. This means that the system's transport mainly proceeds along the boundaries to which the currents are exponentially localized.

Meanwhile, topological protection is also discussed in classical systems such as electrical circuits \cite{imhof2018topolectrical}, acoustics \cite{liu2017pseudospins, ni2017topological}, and photonics \cite{xiao2014surface, yuce2019topological}, as well as active matter \cite{sone2019anomalous, shankar2022topological}, game theory \cite{knebel2020topological, yoshida2021chiral, bai2025topological} and mechanical networks \cite{kane2014topological, lubensky2015phonons, paulose2015topological}. Artificial systems such as metamaterials can be designed in view of stable emergent functionalities. 

It should be noticed that identifying topological protection in the above sense is more specific and different from analyzing the important role that in general the network topology (in the sense of network connectivity) plays for the dynamics; see, for example, ref.s~\cite{szolnoki2009topology, helbing2010evolutionary, perc2013evolutionary, szolnoki2014cyclic} in the context of game theory. 

Less explored is the possible role of topological protection in the context of natural biological systems. In view of the ubiquitous stochasticity  in biochemical reactions, their fluctuations in space and time, the well functioning of genetic, protein, cellular, or neuronal  networks comes as a surprise, and to date it is not well understood in all its facets. In ref.~\cite{tang2021topology} two-dimensional stochastic networks are considered that consist of minimal motifs which represent out-of-equilibrium cycles at the molecular scale, but lead to chiral edge currents, possibly associated with long temporal and spatial scales. More concretely, in ref.~\cite{zheng2024topological} the authors consider a model for KaiC, which regulates the circadian rhythm in cyanobacteria, and demonstrate topological protection of coherent edge oscillations. In contrast to other KaiC models, the model of ref.~\cite{tang2021topology} displays the interesting (since counterintuitive) option of a parameter regime with increased precision at reduced cost.

In general, possible applications of edge activities in the biological context  are in providing stable clocks, mass transport, growth processes, or synchronization.
Topologically protected synchronization \cite{sone2022topological} can be achieved, when nonlinear on-site dynamics of Stuart-Landau oscillators are linearly coupled  by interaction Hamiltonians with topological features in momentum space. In ref.~\cite{qi2006topological},  two choices  for the interaction Hamiltonian are non-Hermitian, both inspired by Chern-insulator models, but differently combined, and one choice is Hermitian, resulting from the first choice by a unitary transformation \cite{sone2022topological}. The resulting dynamics are chaotic in the bulk and synchronized oscillations along the edges of a two-dimensional grid. In particular, as result  of the nonlinear on-site dynamics, extra topological edge modes may emerge, termed exceptional edge modes (see also \cite{sone2020exceptional}). Possible applications are the generation of desired topological synchronization patterns and the detection of disordered oscillations in artificial systems.

In this paper, we consider on-site dynamics of Stuart-Landau oscillators, as well as two other paradigmatic nonlinear models of biochemical oscillations, the activator-repressor model \cite{nakajima2005reconstitution, rust2007ordered, goldbeter1991minimal, pomerening2005systems, danino2010synchronized}, and the brusselator model \cite{nicolis1977prigogine} (for other coupled versions see \cite{biancalani2010stochastic, kolinichenko2020multistability, osipov2000stochastic, pietras2019network}). Here, in view of having the chance of nontrivial topological effects, we choose  the linear coupling of these on-site models  according to the Su-Schrieffer-Heeger (SSH) model  on a two-dimensional square lattice \cite{su1979solitons, lieu2018topological}. The SSH-model is one of the simplest models to study a topological phase transition between topologically trivial and non-trivial phases, which are characterized by Zak phases. In the SSH-model, the phase transition occurs when the ratio between inter-and intracell electron hopping is tuned. Once intercell hopping becomes larger  than intracell hopping, the Zak phase becomes nonzero, protected edge states then result from a mechanism that is termed bulk-boundary correspondence \cite{hasan2010colloquium, qi2011topological}. This correspondence relates robust edge states to bulk topological invariants.\\
The bulk-boundary correspondence was first established in Hermitian systems. In non-Hermitian systems it need not hold in its usual form of Bloch band theory and topological invariants of Bloch bands \cite{xiong2018does, xiao2020non}. \\
Our effective Hamiltonian, derived from linearized non-linear on-site dynamics and linear SSH-coupling of individual units is non-Hermitian. Therefore, from the perspective of the SSH-model, we are interested in the effect of additional on-site dynamics on the occurrence of topologically protected  edge states and the mechanism behind their protection. For the arrangement on the grid we imitate the SSH-coupling with intracell and intercell couplings of different strength. On a square lattice each cell is composed of four coupled nonlinear units with their own on-site dynamics (note that ``site" of ``on-site'' refers to a single node on the grid, not to a grid cell).\\
As we shall see, from a certain ratio of intra-to intercell couplings on, we find strong oscillatory activity restricted to the boundaries, while units in the bulk approach oscillation-death (OD-)states \cite{koseska2013oscillation}. This clear spatial separation between incoherent oscillatory behavior at the edges and quiescence in the bulk bears resemblance to chimera states — dynamical regimes where coherent and incoherent domains coexist within a single system. In our context, the coexistence of rhythmic edge activity and bulk oscillation death can be viewed as a form of frequency chimera-like states. Such states have been reported across diverse systems, including adaptive neuronal networks \cite{wang2020chimeras} and coupled oscillators \cite{faghani2018effects,wei2018nonstationary,mehrabbeik2021synchronization,parastesh2019chimera,hussain2021synchronization}, where they have been analyzed using tools such as bivariate local order parameter decomposition \cite{parastesh2020detecting}.\\
 In general, apart from certain coupling intervals, the oscillations along the boundaries are not synchronized. The Zak phases of the dispersion bands for the bulk system approach a nonzero value, indicating the topological origin of the oscillatory edge states. They vanish in the topological trivial phase which occurs for OD-states in the bulk and along the edge.\\
For one of the different  on-site dynamics we also find that a single switch of an appropriate system parameter can interchange the ``frozen" and the oscillatory areas on the grid. In general it is the ratio between weak and strong couplings that decides about the value of the Zak phases.\\
The paper is organized as follows. In Sec.\ \eqref{sec2}, we define the on-site dynamics as well as the lattice construction, inspired by the SSH-model.  
Section \eqref{sec3} contains the phenomenology obtained from numerical simulations for Stuart-Landau oscillators. We discuss trivial edge effects as well as the weak-coupling limit which provides a heuristic understanding of edge-localized oscillations, discussed in more detail for Stuart-Landau oscillations including their robustness. In Sec.\ \eqref{sec4}  we derive an effective Hamiltonian in momentum space and analyze its symmetries at the example of Stuart-Landau oscillators. We calculate the Zak phases to trace back the topological protection of the edge states. Section \eqref{sec5} presents results from numerical simulations with different on-site dynamics. Section \eqref{sec6}  contains the summary and conclusions. The appendices contain  details of the Hamiltonian's symmetries and their explicit form for the coupled  activator-repressor model and the brusselators.

\section{The Model}\label{sec2}

\par To design a lattice that supports topological edge states, we begin with a structure inspired by the Su-Schrieffer-Heeger (SSH) model \cite{su1980soliton}. This model describes a system where alternating hopping amplitudes between neighboring sites create conditions for the emergence of topological edge states. Specifically, when the hopping amplitude between different unit cells exceeds that within a single unit cell, the Zak phase becomes nonzero, leading to the appearance of topologically protected edge states due to the bulk-boundary correspondence.  

\begin{figure}[htp]
	\centerline{\includegraphics[width=0.5\textwidth]{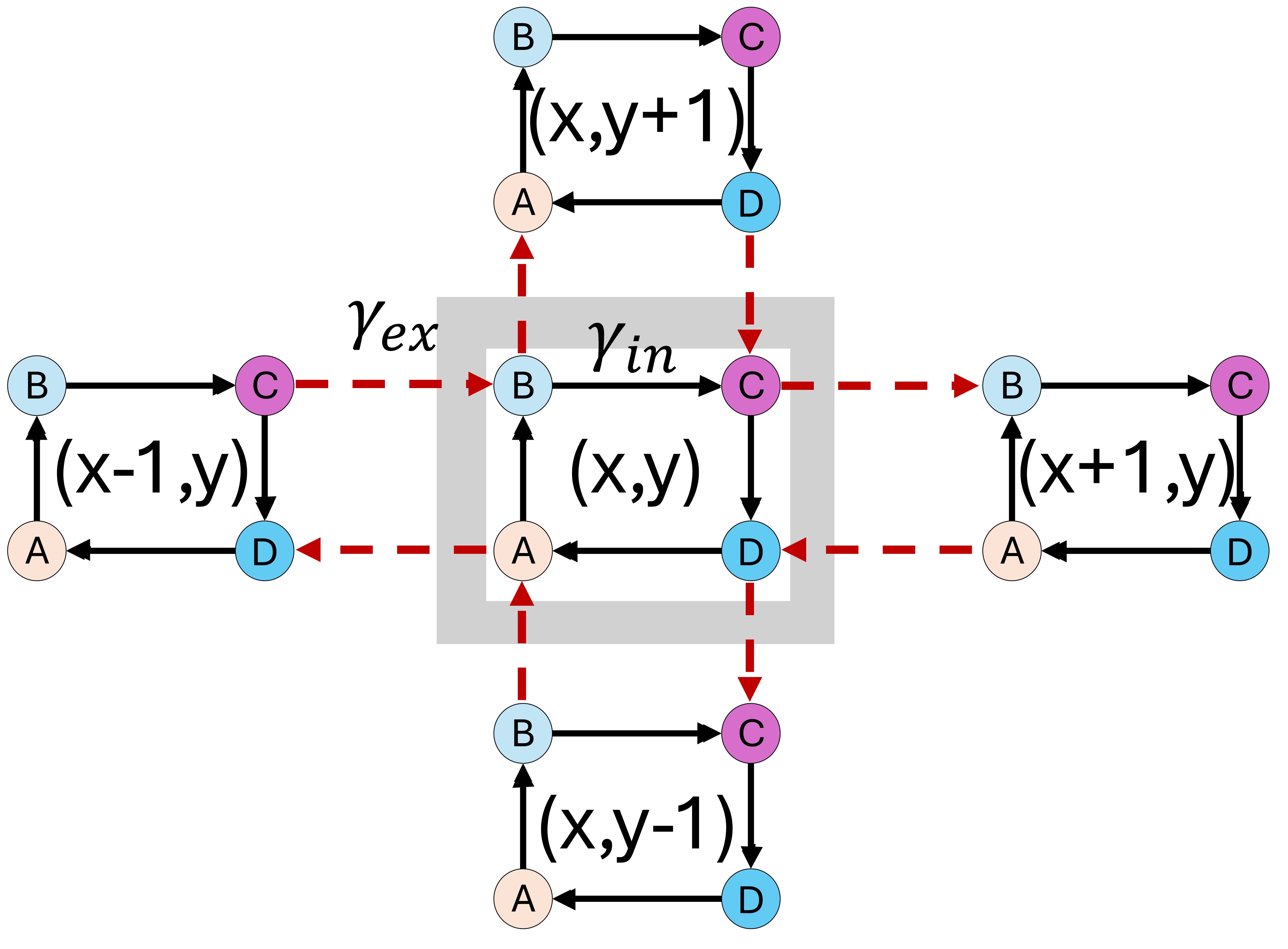}}
	\caption{{\bf Schematic representation of the designed lattice inspired by the SSH model}: A unit cell consists of four oscillators labeled as $A$, $B$, $C$, and $D$, which are 
coupled to each other with $\gamma_{in}$ (black solid lines). One such cell, located at position $(x,y)$,  is connected to its four nearest neighbors at positions $(x+1,y)$, $(x-1,y)$, $(x,y+1)$, and $(x,y-1)$ through 
coupling $\gamma_{ex}$, represented by red dashed lines.
	}
	\label{Picture1}
\end{figure}

Following this idea, we apply the same coupling approach to a system of nonlinear oscillators. The schematic in Fig.\ \eqref{Picture1} illustrates our designed geometry, where a single unit cell consists of four  coupled oscillators, in our case strongly coupled of strength $-s=\gamma_{in}$(black solid lines). The weak coupling between different cells, represented by red dashed lines, has a strength of $-w =\gamma_{ex}$. Each unit cell in the bulk interacts with its four neighboring cells through weak coupling.  

Note that in comparison to the SSH-model, the role of strong and weak couplings here is interchanged with respect to the identification of a unit cell. This is for convenience. The reason is that the stationary state of a strongly coupled unit cell almost coincides with the stationary state of the entire grid in the limit of weak external couplings in the bulk and for open boundary conditions. We make use of this stationary state in the linear stability analysis. (When the weak and strong couplings are then interchanged, the topological characteristic will also change as in the SSH-model, see Sec.\ \eqref{sec4}). In our case, the distinction of unit cells, made of four internally coupled oscillators, and coupled externally to their nearest neighbor cells, looks first artificial: The on-site dynamics are the same at every site on the grid  and the couplings within and between the cells are physically also of the same type, just differing by their strength, but this makes the effect we are interested in.

The dynamics of each individual $i$-th oscillator on the lattice are governed by the following equations:  
\begin{equation}
	\dot{X}_i = F_i(X_i, Y_i), \quad \dot{Y}_i = G_i(X_i, Y_i).
	\label{eq:XY_system}
\end{equation}
These equations describe the time evolution of the state variables \(X_i\) and \(Y_i\), where \(F_i\) and \(G_i\) represent the underlying continuously differentiable nonlinear functions governing the on-site dynamics.  

We consider a lattice consisting of a total of \(L = L_x \times L_y\) units arranged on a two-dimensional grid. To achieve the desired state of topologically protected edge oscillations, both the number of rows, \(L_x\), and the number of columns, \(L_y\), must be even. In the later sections, we will discuss why an odd number of rows or columns prevents the emergence of this state (see Sec.\ \eqref{subsec32}). The dynamical equation of the \(i\)-th oscillator $(i=1,2,3,\cdots,L)$ in the coupled system is given by 
\begin{equation}
	\begin{aligned}
		\dot{X}_i &= F_i(X_i, Y_i) - \sum_{j=1}^{L} \mathcal{A}_{ij} X_j, \\  
		\dot{Y}_i &= G_i(X_i, Y_i),
	\end{aligned}
	\label{eq:coupled_system}
\end{equation}
where $\mathcal{A}_{ij}\in\{-s,-w,g_i, 0\}$, and $s,w,g_i\ge0$, represents the components of the coupling matrix that defines the interaction between the oscillators that we choose in general to be directed. 
The diagonal elements, \(\mathcal{A}_{ii} = g_i\), represent additional linear self-coupling terms of the oscillators, usually not included in the standard form of the on-site dynamics, which we allow to be different from zero in view of their possible effect on the stationary state of the coupled system.  Unless explicitly mentioned otherwise, the parameter $g_i$ is assigned a fixed non-negative uniform value over the grid  and remains unchanged throughout the analysis. In the context of parameter mismatch, $g_i$ will become site-dependent.
 We choose open boundary conditions to observe edge effects and periodic boundary conditions to determine bulk physics.

For the on-site dynamics, we first consider Stuart-Landau (SL) oscillators \cite{kuramoto2003chemical} as normal form of Andronov-Hopf bifurcations,  given by  
\begin{equation} \label{eqL}
	\begin{aligned}
		F(X_i, Y_i) &= \left[\alpha_i - \beta_i \left(X_i^2 + Y_i^2\right)\right] X_i - \omega_i Y_i, \\
		G(X_i, Y_i) &= \left[\alpha_i - \beta_i \left(X_i^2 + Y_i^2\right)\right] Y_i + \omega_i X_i.
	\end{aligned}
\end{equation}
Each isolated $i$-th oscillator can exhibit a self-sustained limit cycle with frequency $\omega_i$ and amplitude $\sqrt{\dfrac{\alpha_i}{\beta_i}}$, where $\dfrac{\alpha_i}{\beta_i} > 0$ for each $i=1,2,3,\cdots,L$. 

In less detail we will present  results for two other nonlinear on-site dynamics, which individually serve as prototypes for many applications in biological systems. These are activator-repressor units \cite{krishna2009frustrated, kaluza2010role} and brusselators \cite{nicolis1977prigogine}.

Unlike SL-oscillators, which exhibit limit-cycle behavior governed by a simple nonlinearity, the activator-repressor system includes a feedback mechanism. The dynamics of the \( i \)-th oscillator are given by:  
\begin{equation} \label{BFU}
	\begin{aligned}
		\dot{X}_i &= F(X_i, Y_i) = \frac{\alpha_i}{1 + \frac{Y_i}{K}} \cdot \frac{b + X_i^2}{1 + X_i^2} - X_i, \\
		\dot{Y}_i &= G(X_i, Y_i) = \gamma (X_i - Y_i).
	\end{aligned}
\end{equation}  
Here, \( X_i \) and \( Y_i \) represent the activator and repressor concentrations, respectively. The parameter $\gamma$ is the ratio of the half-life of $X_i$ to that of $Y_i$, $K$ sets the strength of repression of $X_i$ via $Y_i$, and $b$ determines the basal expression level of $X_i$. The first equation describes the nonlinear regulation of the activator, incorporating saturating feedback through the denominator and an additional nonlinear term in the numerator. The second equation ensures coupling between the two variables, with \( \gamma \) controlling the time scale of the relaxation dynamics of the repressor. 
 
The model, considered in \cite{krishna2009frustrated}, served as an effective description of a genetic network, but the motif is quite generic and also found in neural networks.  In the deterministic limit (of infinitely many molecules), such a unit is known to show both excitable and oscillatory dynamics, depending on the choice of parameters \cite{kaluza2010role}. When activator-repressor units are repressively and nonlinearly coupled, a variety of states was found in \cite{labavic2014networks}, ranging from a toggle switch for two repressively  coupled  units to rich collective coherent behavior of individual fixed-point and limit-cycle behavior. There, the option showed up to control the duration of oscillations by monotonic variation of a single bifurcation parameter that may have an experimental counterpart. Thus, a large system of coupled activator-repressor units shows promising features for biological applications.
 
Differently from this previous work \cite{labavic2014networks}, here, by introducing attractive and linear couplings between these units with alternating strength, we want to explore whether  oscillations can be restricted to edges as for SL-oscillators. 

The brusselator  \cite{lefever1971chemical} is one of the simplest chemical models  for autocatalytic reactions such as the Belousov–Zhabotinsky reaction. When coupled via diffusion, the brusselator undergoes a Turing bifurcation resulting in the formation of Turing patterns like spots, stripes, and spirals in two spatial dimensions.

The dynamics of the $i$-th oscillator of the brusselator model are given by:
\begin{equation} \label{Brusselator}
	\begin{aligned}
		\dot{X}_i &= F(X_i, Y_i) = a_i - (b_i+1)X_i + X_i^2 Y_i, \\
		\dot{Y}_i &= G(X_i, Y_i) = b_i X_i - X_i^2 Y_i.
	\end{aligned}
\end{equation}
The state variables represent concentrations of chemical species, which must remain non-negative. Individually, also the brusselator exhibits limit-cycle oscillations as long as $b_i>1+a_i^2$ \cite{lefever1971chemical}. 

\section{Phenomenology}\label{sec3}
First we present results for coupled SL-oscillators as function of varying (in particular common)  coupling strength, chosen to be the same between external and internal  couplings, followed by a focus on the weak-coupling limit in the extreme case of vanishing weak coupling. Turning on weak, but non-vanishing coupling, we observe edge states and investigate the robustness of the observed edge-oscillations.

At a glance, in most cases we observe either oscillations (in general not synchronized), or OD-states. As the name suggests, OD-states result from strongly coupled oscillators which block each other towards oscillation death. Even when the  oscillation amplitudes are not exactly zero, but smaller than a model-dependent low threshold while approaching fixed-point values, we term theses states OD-states. The oscillations may be restricted to the boundaries, where they are completely incoherent or form chimera-like states with respect to their frequencies, or oscillations extend over the entire lattice, or all units are in OD-states, as we shall see. We use periodic boundary conditions to determine bulk behavior and open boundary conditions to simulate edge oscillations. For the numerical integration we use the Runge-Kutta-Fehlberg method, usually  for $t = 10000$ time steps with an integration step size of $\delta t = 0.01$.
Unless otherwise stated, a standard choice of parameters is $\beta_i = 1.0$, $\alpha_i = \alpha = 1.0$, and $\omega_i = \omega = 3.0$, for which each SL oscillator exhibits sinusoidal-like oscillations within the range $[-1,1]$. Therefore, we set the initial conditions for each oscillator randomly within the region $[-1,1] \times [-1,1]$. 

In Fig.\ \eqref{Picture2} we show a minimal grid of size $4\times 4$ with only one unit cell in the bulk, but two interchanged assignments of strong and weak couplings in (a) and (b), which we will use for discussions in the following. In anticipation of later results, version (a) leads to edge-localized oscillations with strong couplings within a unit cell, or version (b) to OD-states everywhere with weak intracell-coupling.
	\begin{figure}[htp]
	\centerline{\includegraphics[width=0.5\textwidth]{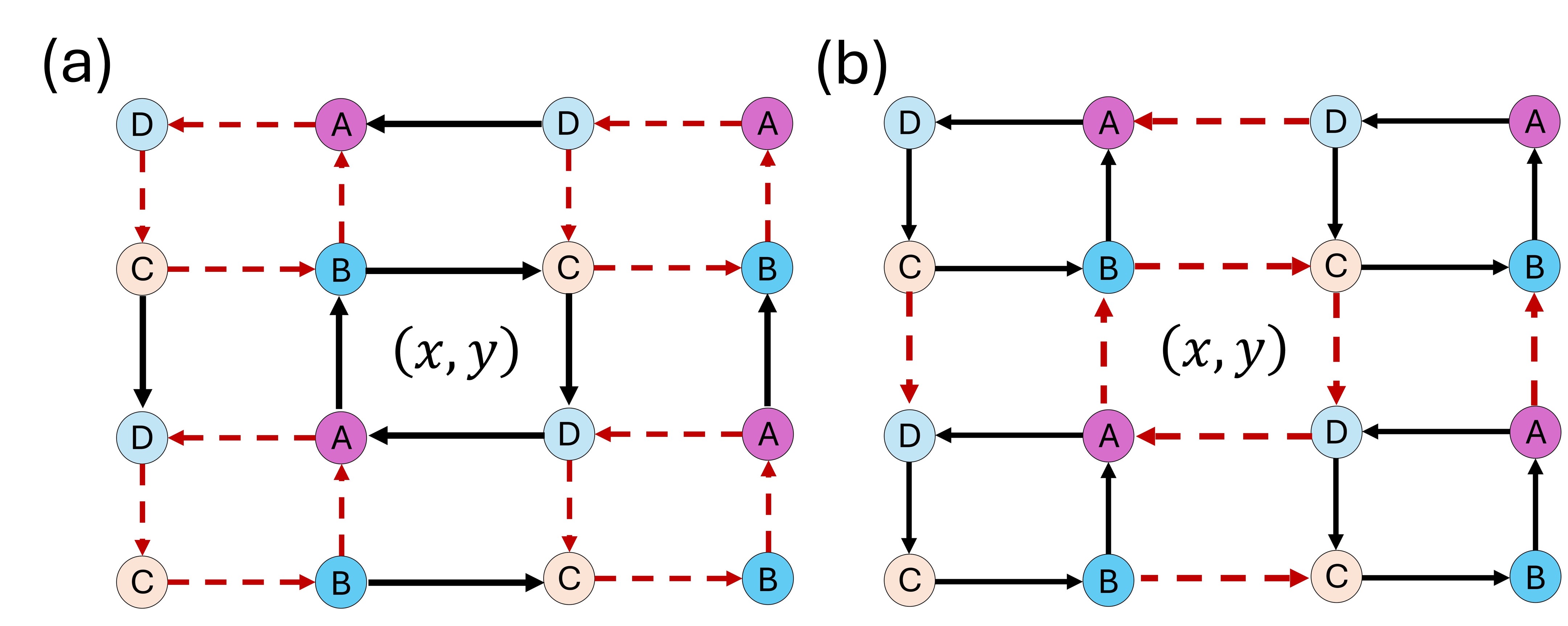}}
	\caption{{\bf Designed lattice structure}: This figure illustrates the coupling scheme inspired by the SSH model on a \(4 \times 4\) lattice. The red dashed lines represent the weak coupling of strength \(-w\), while the solid black lines indicate the strong coupling of strength \(-s\). The four oscillators, located in the interior of the lattice, are referred to as bulk oscillators. The remaining oscillators on the boundary of the lattice are called edge oscillators. Panel (a) shows a unit cell with strong internal coupling, weakly coupled to the boundaries, (b) the case with inverted coupling strengths so that the bulk is strongly coupled to the boundaries. 
	}
	\label{Picture2}
\end{figure}

\subsection{Trivial boundary effects}\label{subsec31}

Since the focus in this paper is on differences between edge and bulk dynamics, one should identify trivial topological effects which merely result from the fact that by definition boundary elements have less interactions to their nearest  neighbors than elements in the bulk. 
For $w=s=0.001$ and other parameters chosen as $g_i=0,\beta_i=1.0,\alpha_i=1.0, \omega_i=3.0$, all units oscillate incoherently (not displayed). For $w=s=8.0$ and other parameters as before, all units settle into different OD-states with values that depend on the initial conditions. Now, varying $w$ for $s=8.0$ to other values, we obtain Fig.\ 3(a) for  oscillators located along the edge of the grid (edge oscillators) and (b) for  oscillators located in the bulk of the grid (bulk oscillators).

\begin{figure}[htp]
\includegraphics[width=0.5\textwidth]{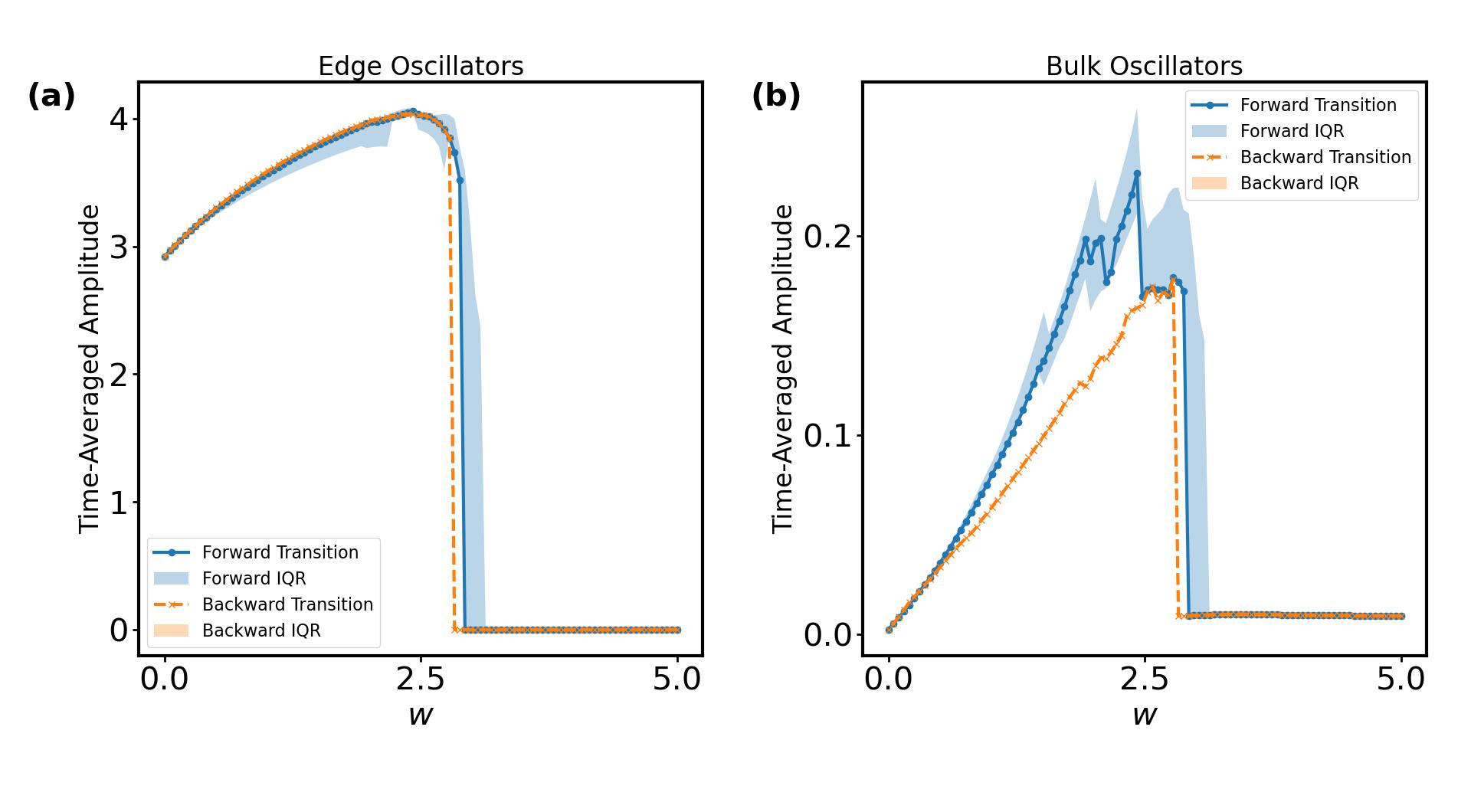}
	\caption{{\bf Oscillations in the bulk are damped as compared to those on the edge.}
 Time-averaged amplitude of oscillators as a function of coupling strength \( w \) in a \( 10 \times 10 \) lattice of coupled SL-oscillators, (a)  for edge oscillators, (b) for bulk ones. 
 Parameters are: \( g = 0 \), \( \beta_i = 1.0 \), \( \alpha_i = \alpha = 1.0 \), $s=8.0$, and \( \omega_i = \omega = 3.0 \). For further explanations see the text.	
	}
	\label{Picture14}
\end{figure}
	
To generate this figure, we first identify  edge and bulk oscillators for a grid of $10\times 10$ nodes. The time-averaged amplitude \cite{dixit2021emergent,sharma2012amplitude} of each oscillator is determined by calculating the difference between its maximum and minimum oscillation values for a given \( w \), followed by averaging this difference over all oscillators within the respective group.
The time-averaged amplitude is defined as:
	\begin{equation}\label{Time-average amplitude}
			\frac{1}{N_{\text{edge}}} \sum_{j=1}^{N_{\text{edge}}} \left[ \langle X^{\max}_j  \rangle_t - \langle X^{\min}_j  \rangle_t \right], 
\end{equation}
	analogously for the bulk oscillators, with $N_{\text{edge}}$ being the number of oscillators on the edge.
	Here, \( \langle \cdots \rangle_t \) denotes the time average taken over a sufficiently long duration.
	
For both forward and backward transitions, the simulations begin with random initial conditions sampled from \([-1,1] \times [-1,1]\). The final state of the system \eqref{eq:coupled_system} at a given \( w \) is then used as the initial condition for the subsequent iteration. This  is repeated across five independent runs, each starting with different initial conditions. The median and interquartile range (IQR, spanning the $25$th to $75$th percentile) of the oscillation amplitude are plotted in  Fig.\ \eqref{Picture14} to illustrate variations across different realizations. 
The main difference between edge and bulk units due to trivial boundary effects shows up in the magnitude of their oscillations. The amplitude of edge oscillators is an order of magnitude larger than that of oscillators in the bulk, but from a certain coupling strength on ($w=2.83$), both sets of units simultaneously go to OD-states. Due to the choice of scale, the figure does not show the actual onset of bulk oscillations, after being in an OD-state for sufficiently small values of $w$ such as $w=0.001$. In view of that, it is of interest, how large the weak coupling may become to keep the bulk units in OD-states, while sharply restricting the oscillations to the edge, before both types oscillate together and then together go to a resting state in an explosive transition \cite{boccaletti2016explosive,leyva2012explosive,dixit2021dynamic}. 

\subsection{The weak-coupling limit}\label{subsec32}
It is instructive to first consider the weak-coupling limit $w\rightarrow 0$ of a grid that is otherwise designed according to the SSH-model or Fig.\ \eqref{Picture2}. Weak-coupling everywhere is not sufficient to restrict the oscillations to the boundaries, as mentioned before.
	
Consider Fig.\ \eqref{Picture2}(a), which consists of four bulk oscillators  and twelve edge oscillators. When $w=0$, the system \eqref{eq:coupled_system} consists of three distinct types of oscillators. Firstly, the four corner oscillators remain completely isolated and oscillate within the range $[-1,1]$. Secondly, the four bulk oscillators  are strongly coupled with $s=8.0$ and undergo oscillation-death, converging to identical  fixed-point values which turn out to be locally stable. Using Newton's method with suitable positive initial conditions, we compute these steady states, which are found to be $(2.5202,1.1377)$. The eigenvalues of the Jacobian at these steady states, shown in Fig.\ \eqref{Picture10}, confirm their local stability. Notably, if the Newton method is initialized with negative coordinates, the four bulk oscillators instead converge to the mirror steady state $(-2.5202,-1.1377)$, although the eigenvalue structure remains unchanged, as shown in Fig.\ \eqref{Picture10}. The presence of negative real parts in all eigenvalues ensures that these OD-states are locally stable. 
		\begin{figure}[htp]
		\centerline{\includegraphics[width=0.4\textwidth]{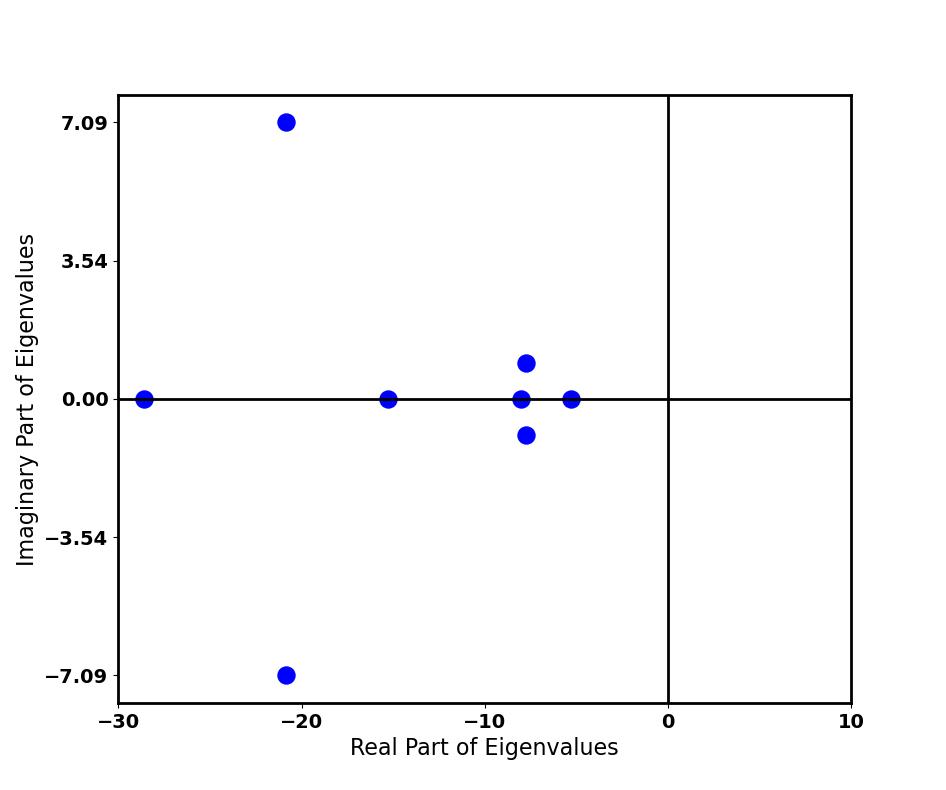}}
		\caption{{\bf Stability analysis of a strongly coupled unit cell of SL-oscillators}: Plotted are the eight eigenvalues of the Jacobian matrix evaluated at the OD-states.  The negative real parts of all eigenvalues confirm the local stability of these steady states, indicating that the oscillators cease to oscillate. Parameters: $s=8.0$, \(\beta_i = 1.0\), $g=0$,  \(\alpha_i = \alpha = 1.0\), and \(\omega_i = \omega = 3.0\).
		}
		\label{Picture10}
	\end{figure}
	\par The remaining oscillators form four pairs of unidirectionally coupled units. Within each pair, one oscillator exhibits low-amplitude oscillations confined to $[-1,1]$, while the other one oscillates with a higher amplitude in the range $[-2,2]$, the asymmetry in amplitudes is due to the directed coupling. 
	Thus, in the decoupled case with respect to the weak links, we find oscillations only along the boundary.
For a coupling assignment according to Fig.\ \eqref{Picture2}(b), all units go to OD-states.

\par In the weak-coupling limit it becomes also obvious why the same arrangement of couplings on a lattice as in Fig.\ \eqref{Picture2}(a) with an odd number of rows or columns fails to sustain edge oscillations, or, equivalently, why an arrangement according to Fig.\ \eqref{Picture2}(b) fails as well. Consider again Fig.\ \eqref{Picture2}(a). We modify this structure by removing the last column and we set \(w=0\). As a result, the edge oscillators now include nodes, which become part of a set of four strongly coupled oscillators with $s=8$. Consequently, these oscillators converge to OD-states, as discussed earlier. 
This explains why maintaining an even number of rows and columns for arbitrary finite lattice size in an arrangement according to Fig.\ \eqref{Picture2}(a) is essential for the emergence of edge oscillations.

\subsection{Edge-localized oscillations}\label{subsec33}
\par From Fig.\ \eqref{Picture2}(a) and vanishing coupling we know how to restrict oscillations to elements of the edges of the grid. How long can we maintain this restriction when the weak coupling is turned on, that is, towards which values of the weak coupling? Certainly, when the weak coupling gets too strong in relation to the strong one, there is no such localization of oscillations as we have seen in Fig.\ \eqref{Picture14}.
For the standard choice of initial conditions and of other parameter values, we next set $s = 8.0$ and $w = 0.001$, and consider a $10 \times 10$ lattice of coupled SL-units. Additionally, we turn off the self-feedback by setting $g_i = g = 0$. 

\begin{figure}[htp]
	\centerline{\includegraphics[width=0.5\textwidth]{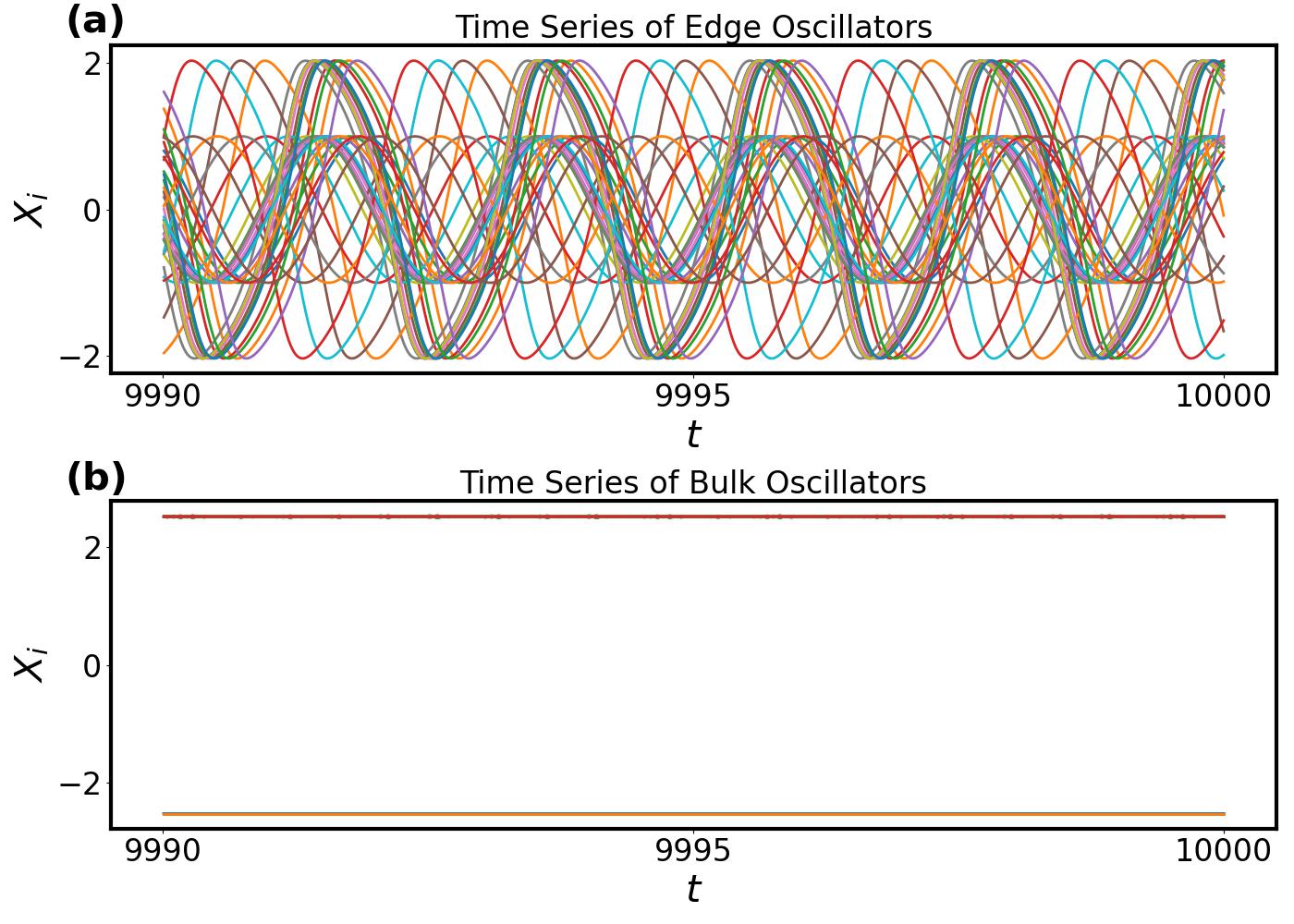}}
	\caption{{\bf Nontrivial edge oscillations on a \(10 \times 10\) lattice with coupled SL units}:  Edge oscillators sustain their oscillations, while the bulk oscillators converge to OD-states, despite all oscillators sharing identical parameters: \(\beta_i = 1.0\), \(\alpha_i = \alpha = 1.0\), and \(\omega_i = \omega = 3.0\).  Initial conditions for each oscillator are randomly chosen from \([0,1] \times [0,1]\). The coupling parameters are set as \(s=8.0\), \(w=0.001\), and \(g=0\).
	 }
	\label{Picture3}
\end{figure}

\par The $64$ bulk oscillators converge to two-cluster OD-states, where these two clusters are symmetrically distributed with respect to the origin (Fig.\ \eqref{Picture3}(b)). Depending on the initial conditions, the bulk oscillators join one of these clusters.
On the other hand, the $36$ edge oscillators exhibit oscillations with varying amplitudes, as shown in Fig.\ \eqref{Picture3}(a). While each isolated SL-oscillator, given the chosen parameter values, oscillates within the range $[-1,1]$, some of the edge oscillators remain within this range, whereas others exhibit higher-amplitude oscillations, reaching up to $[-2,2]$. This behavior is observed for any choice of random initial conditions chosen within the range $[-1,1] \times [-1,1]$ as well as for larger lattice sizes  with an even number of rows, $L_x$, and an even number of columns, $L_y$ and more cells in the bulk. 

In particular, edge oscillations are also found in the absence of any self-feedback term, i.e., with $g_i = g = 0$. The observed findings are consistent with the fact that coupled SL-oscillators may converge to OD-states in the presence of symmetry-breaking coupling, which serves as a necessary condition for the emergence of such coupling-dependent OD-states \cite{koseska2013oscillation,koseska2013transition}. Here, the symmetry-breaking   refers to couplings only in the $X_i$-components rather than in both $X_i$ and $Y_i$.

\begin{figure}[htp]
	\centerline{\includegraphics[width=0.5\textwidth]{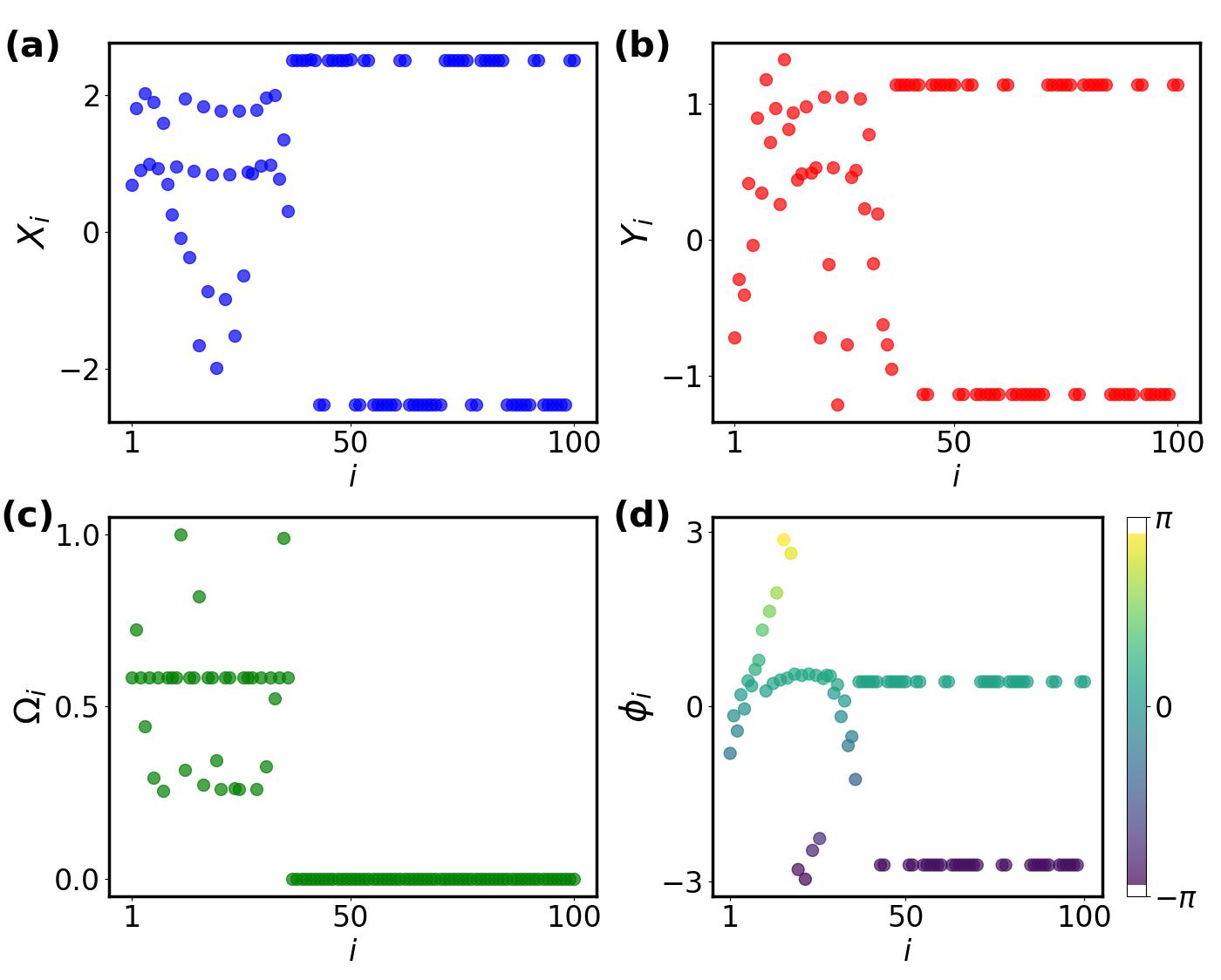}}
	\caption{{\bf Frequency chimera in edge oscillators and OD-states in the bulk}: Snapshots of the state variables at a particular time after a sufficiently long initial transient: (a) $X_i$, (b) $Y_i$, (c) normalized frequency $\Omega_i$, and (d) phase $\phi_i$. The lattice consists of a $10 \times 10$ grid of coupled SL-oscillators, as in Fig.\ \eqref{Picture3}. The $x$-axis is reordered such that the edge-oscillator indices are plotted first, followed by the bulk oscillators. For further explanations see the text.
}
\label{Picture4}
\end{figure}

\par To further characterize the emergent state of the system, we present snapshots of the oscillators' states, including \( X_i \), \( Y_i \), the normalized frequency values \( \Omega_i \), and the phase of each oscillator, \( \phi_i(t) = \arctan{\left(\frac{Y_i}{X_i}\right)} \), at a specific time step after the system has evolved sufficiently beyond its initial transient, as shown in Fig.\ \eqref{Picture4}. The normalized frequency values \( \Omega_i \) for each oscillator are determined by first calculating the phase angles \( \phi_i(t) \) at each time step. 
The phase velocity, which represents the instantaneous frequency of each oscillator, is then computed as the numerical derivative of the phase with respect to time. These instantaneous frequencies are averaged over a sufficiently long time window to obtain the mean frequency for each oscillator. The reason why we average over the frequencies is to avoid misinterpretation: an oscillator that stays quiet for a while might look like it is in an OD-state if only a snapshot is considered.
To normalize the frequencies, the maximum absolute frequency across all oscillators is identified, and each individual frequency is scaled by this maximum value. As a result, the normalized frequencies lie between $0$ and $1$ in Fig.\ \eqref{Picture4} (c). 

\par In each of the subfigures of Fig.~\eqref{Picture4}, we reorder the x-axis to display the oscillator indices in such a way that the edge units are plotted first, followed by the bulk units. Subfigures (a) and (b) in Fig.~\eqref{Picture4} are consistent with the results observed in Fig.~\eqref{Picture3}, where the edge units (the first $36$ unit indices on the $x$-axis) exhibit random scattering, indicating incoherent oscillations, while the bulk units (the last $64$  indices) converge to OD-states or show negligible oscillations around these states. These bulk units settle into two coherent groups depending on their initial conditions. Subfigure (c) further highlights the  feature that  a subset of the edge oscillators exhibit the same normalized non-vanishing frequencies, while the remaining edge oscillators have random, nonzero values. This suggests that the edge oscillators follow a frequency chimera-like behavior \cite{parastesh2021chimeras,majhi2019chimera,abrams2004chimera,khaleghi2019chimera}, where one group of edge oscillators exhibits coherent frequencies, while the other group remains incoherent. This chimera-like behavior holds regardless of the initial conditions chosen from the range $[0,1] \times [0,1]$. It is a special chimera-like state that is topologically protected, as we shall explain in Sec.\ \eqref{sec4}. In general, frequency chimera states can emerge from different mechanisms \cite{parastesh2021chimeras}.  Furthermore, in subfigure (d) of Fig.~\eqref{Picture4}, we plot the phases of each oscillator, which reveals that the edge oscillators display incoherent phases, while the bulk oscillators converge into two distinct phase groups, consistently with subfigures (a) and (b).\\
So far we have  considered open boundary conditions in $x$ and $y$ direction for all results presented so far. If we choose periodic boundary conditions in one of these directions, the  oscillations are confined to the edges with open boundary conditions, not further displayed here.\\
Further explorations of parameter changes are summarized in Appendix \eqref{AppendixA}.

\subsection{Robustness of edge-localized oscillations}\label{subsec34}
The mechanism of topological protection adds upon other possible origins of robustness and resilience in complex networks such as structural design, identification of early warning signals, or the devise of adaptive responses \cite{artime2024robustness, mikaberidze2024consensus}. 
In applications,  when robustness against various kind of perturbations is caused by topological protection, the  desired features (here of the localized edge oscillations) are much less sensitive to the choice of parameter values, since topological invariants are responsible for the protection. Therefore, we check the robustness against an inhomogeneous choice of $\alpha_i,\omega_i,\beta_i$ and $g_i$, as well as against structural defects and a bidirectional component in the couplings.

\subsubsection{Robustness against mismatch of parameters}\label{subsub341}
{\bf Varying $\alpha_i$ between different oscillators}. Specifically, we keep the parameters $\beta_i = 1.0$ and $\omega_i = \omega = 3.0$ fixed, while allowing the parameter $\alpha_i$ to be uniformly distributed around the mean value $\alpha = 1.0$ with a distribution width of $\Delta \alpha = 0.5$. Despite this variation, the overall dynamics remain qualitatively unchanged, figure not displayed. 
In this setup, the number of OD-clusters among the bulk units increases compared to the case with identical units. However, the edge oscillators continue to exhibit sustained oscillations and retain their frequency chimera-like behavior. 

\begin{figure}[htp]
	\centerline{\includegraphics[width=0.5\textwidth]{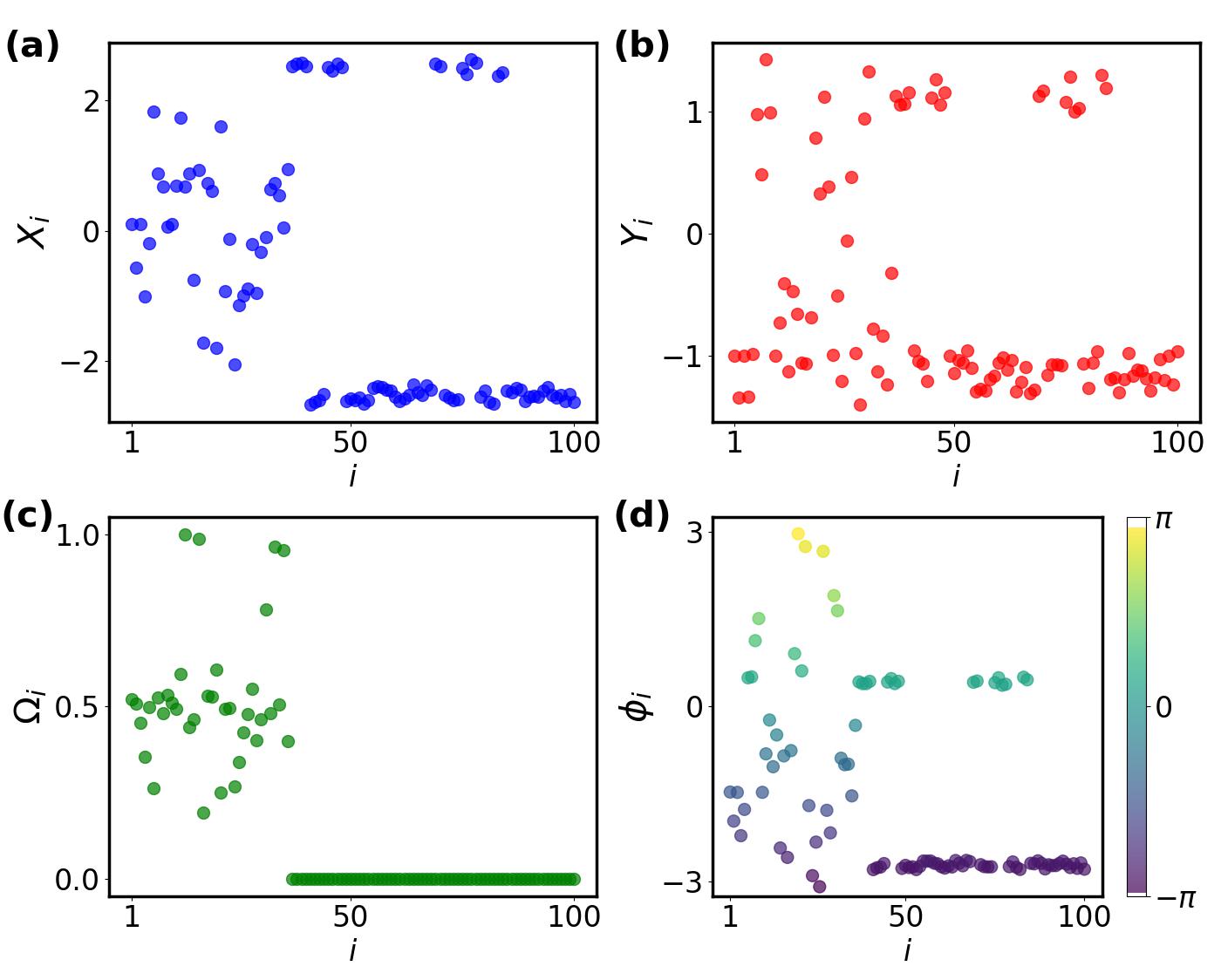}}
	\caption{{\bf Maintenance  of edge-localized oscillations in spite of disruption of frequency chimera in edge oscillators due to inhomogeneous intrinsic frequencies}: We introduce a small parameter mismatch in the intrinsic frequencies $\omega_i$ of the SL oscillators, drawn uniformly around a mean $\omega = 3.0$ with width $\Delta \omega = 0.5$, keeping all other parameters as in Fig.\ \eqref{Picture4}. While edge oscillators continue to oscillate, the frequency disorder disrupts their partial coherence. In contrast, bulk oscillators converge to oscillation death, as indicated by zero normalized frequency $\Omega_i$, though they settle into distinct OD-states. Oscillator indices are reordered to show edge units first. Note that edge and bulk units together still form a frequency chimera-like state.
	}
\label{Picture6}
\end{figure}

{\bf Varying $\omega_i$}. Concretely, we fix the parameters at $\beta_i = 1.0$ and $\alpha_i = \alpha = 1.0$ in Fig.\ \eqref{Picture6}, while introducing parameter mismatch through $\omega_i$, which is uniformly distributed around the mean value $\omega = 3.0$ with a distribution width of $\Delta \omega = 0.5$. In this scenario, the edge oscillators continue to oscillate, whereas the bulk oscillators converge to different OD-states depending on their initial conditions. Notably, the frequency chimera-like state observed earlier among the edge oscillators is now destroyed, as no coherent group can be identified among them (see Fig.~\eqref{Picture6}). With respect to all units, subfigure (c) of Fig.~\eqref{Picture6} illustrates that the normalized frequencies of the edge oscillators remain incoherent with nonzero values, while those of the bulk oscillators are coherent. The bulk oscillators settle into distinct OD-states, resulting in zero normalized frequencies. Although minor oscillations around the OD-states may persist in the bulk units, these oscillations are negligible compared to the amplitude of isolated oscillators. Subfigure (c) also illustrates the emergence of a frequency chimera-like state among all oscillators, where the bulk oscillators form a coherent group, while the edge oscillators remain incoherent.

{\bf Additional linear self-coupling turned on}. We further set all parameter values to \( s = 8.0 \), \( w = 0.001 \), \( \beta_i = 1.0 \), \( \alpha_i = \alpha = 1.0 \), and \( \omega_i = \omega = 3.0 \), and randomly assign each oscillator an additional linear coupling \( g_i \) drawn from the interval \( [0, 1] \). The results remain qualitatively the same: edge oscillators continue to oscillate, while bulk oscillators converge to an OD-state. However, the frequency chimera-like state among the edge oscillators is lost. Despite this, we still observe two distinct groups of oscillators: the edge oscillators are incoherent with nonzero normalized frequencies, while the bulk oscillators are coherent with zero frequencies. In this sense, the system forms a frequency chimera-like state. Unless otherwise stated, the initial conditions for all SL-oscillators in this study are drawn randomly from the interval \( [-1, 1] \times [-1, 1] \).

\subsubsection{Robustness against noise}\label{subsub342}
If we include additive noise to both components $X_i$ and $Y_i$ as random numbers drawn from two different distributions: the standard normal distribution $N(0,1)$ and the uniform distribution $[-1,1]$, the coupled system demonstrates the same behavior as shown in Figs.\ \eqref{Picture3} and  \eqref{Picture4}, as long as the noise strength is less than or equal to $10^{-2}$. We present the results for a fixed noise strength of $10^{-2}$ in Fig.\ \eqref{Picture7} in a different way. 
\begin{figure}[htp]
	\centerline{\includegraphics[width=0.45\textwidth]{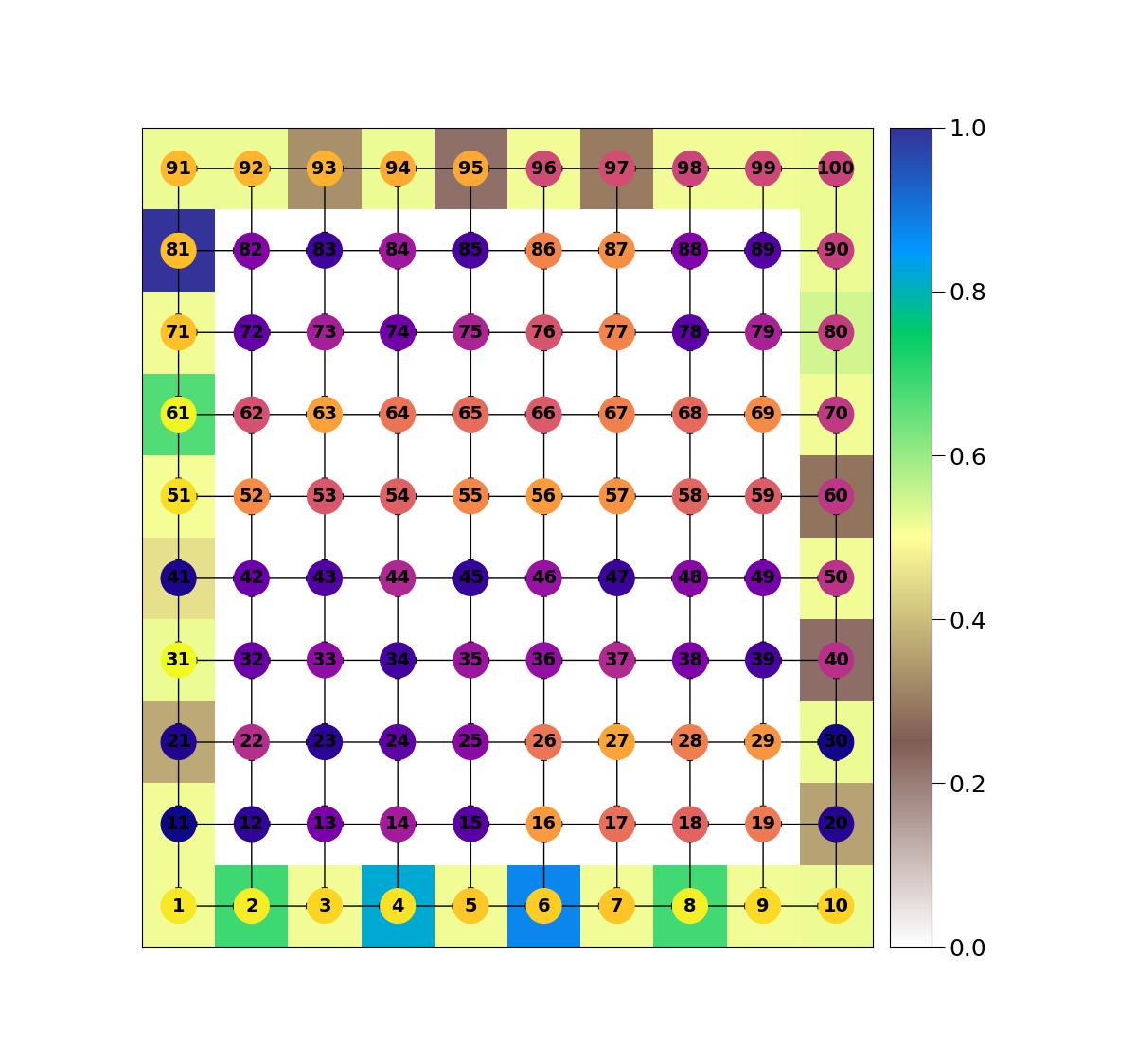}}
	\caption{{\bf Effect of additive noise}:  $10 \times 10$ lattice of SL-oscillators under additive noise on both components with a strength of $10^{-2}$. The nodes are colored according to their time-averaged phases, while the background is colored based on the normalized frequency. The colored background represents the normalized frequency of each oscillator. The parameter values are: \(\beta_i = 1.0\), \(\alpha_i = \alpha = 1.0\), \(\omega_i = \omega = 3.0\), $s = 8.0$, $w = 0.001$, and $g = 0$. The initial conditions are randomly drawn from $[-1, 1] \times [-1, 1]$.
	}
	\label{Picture7}
\end{figure}
In this figure, the time-averaged phases are rounded to six decimal places, and unique phase values are extracted. Each unique phase corresponds to an oscillator, and we assign a color to each node based on its phase. The background is colored according to the normalized frequency. Clearly, the bulk units, which have a zero normalized frequency, exhibit OD-behavior, while the edge oscillators continue oscillating, possessing a nonzero normalized frequency. The colored background in Fig.\ \eqref{Picture7} represents the normalized frequency of each oscillator, providing insight into the variation of intrinsic dynamics across the network. The resulting dynamics remain qualitatively unaltered if additive noise is applied to any one of the component instead of both.

\begin{figure}[htp]
	\centerline{\includegraphics[width=0.4\textwidth]{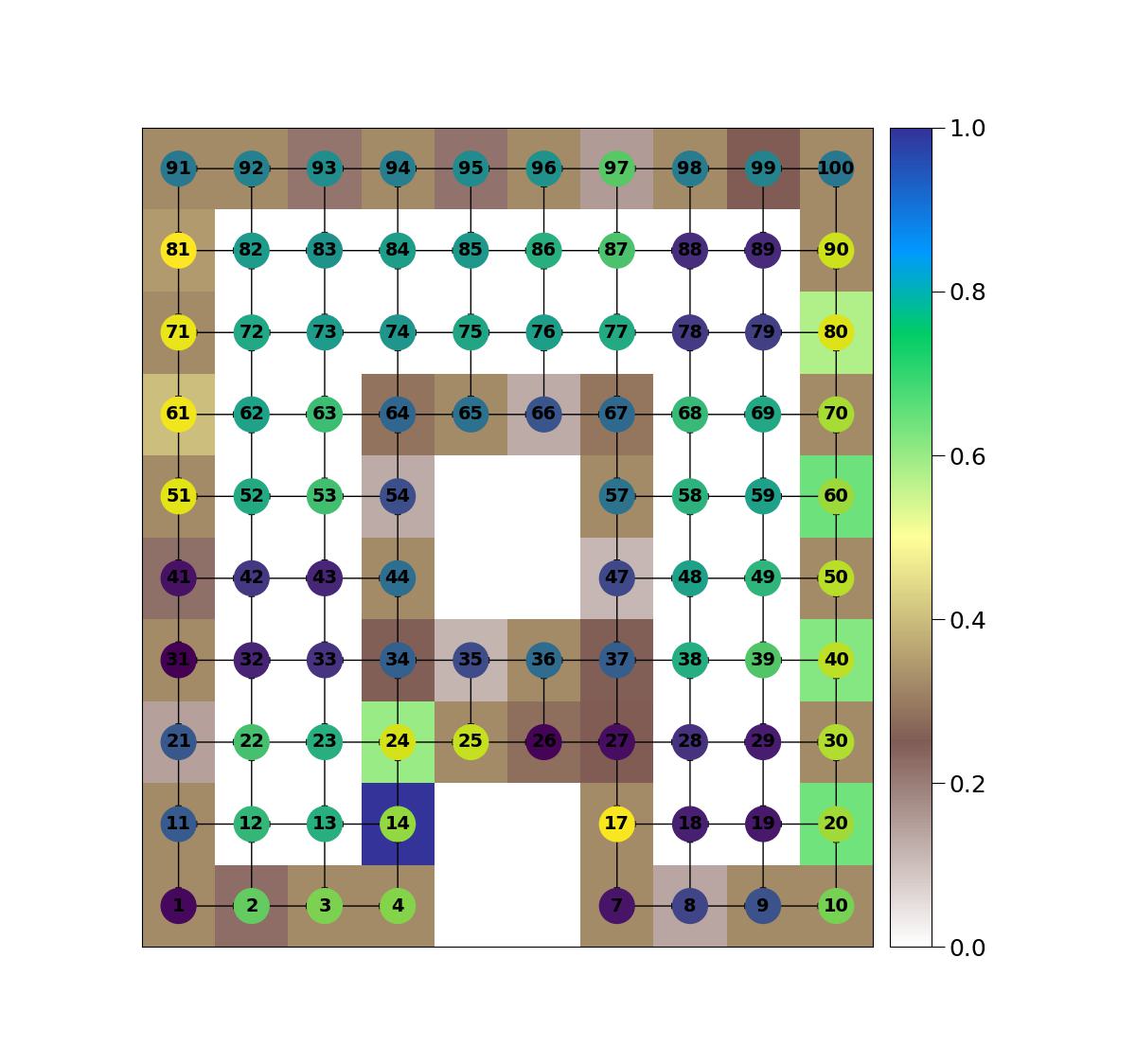}}
	\caption{{\bf Oscillations along a multiply connected edge of a defective lattice}: 
		White areas indicate removed grid points including their associated links, otherwise same color coding as in Fig.\ \eqref{Picture7}. Edge oscillations adapt to the newly created boundaries in the lattice, effectively wrapping around these boundaries. Parameters: $s=8.0$, $w=0.001$, $g=0$, $\beta_i = 1.0$, $\alpha_i = \alpha = 1.0$, and $\omega_i = \omega = 3.0$.
	}
	\label{Picture8}
\end{figure}

\subsubsection{Robustness against structural defects}\label{subsub343}
\par The observed edge oscillations appear to be highly robust against parameter mismatches up to a certain order and additive noise of suitable strength. A natural question arises: what happens if the designed lattice contains inherent defects in the form of missing grid points or a deformed boundary? To explore the impact of such a defected lattice on the system's dynamics, we remove a few oscillators from the bulk 
along with their corresponding connections. Additionally, we delete a few oscillators from the boundary. 

After eliminating these  grid points from the \(10 \times 10\) lattice of coupled SL-oscillators, we numerically integrate the system \eqref{eq:coupled_system} while keeping all other parameters and initial conditions unchanged, as in Fig.\ \eqref{Picture3}. The resulting dynamics is shown in Fig.\ \eqref{Picture8} after a sufficient transient.
The edge oscillators appear to wrap around the defects and continue oscillating with a nonzero normalized frequency. Meanwhile, the bulk oscillators still exhibit a zero normalized frequency, indicating their convergence towards OD-states.

This result suggests to induce defects in a controlled way to achieve oscillating layers within a bulk of resting states at wish. Figure \eqref{Picture110} gives an example for such a possibility.
\begin{figure}[htp]
	\centerline{\includegraphics[width=0.45\textwidth]{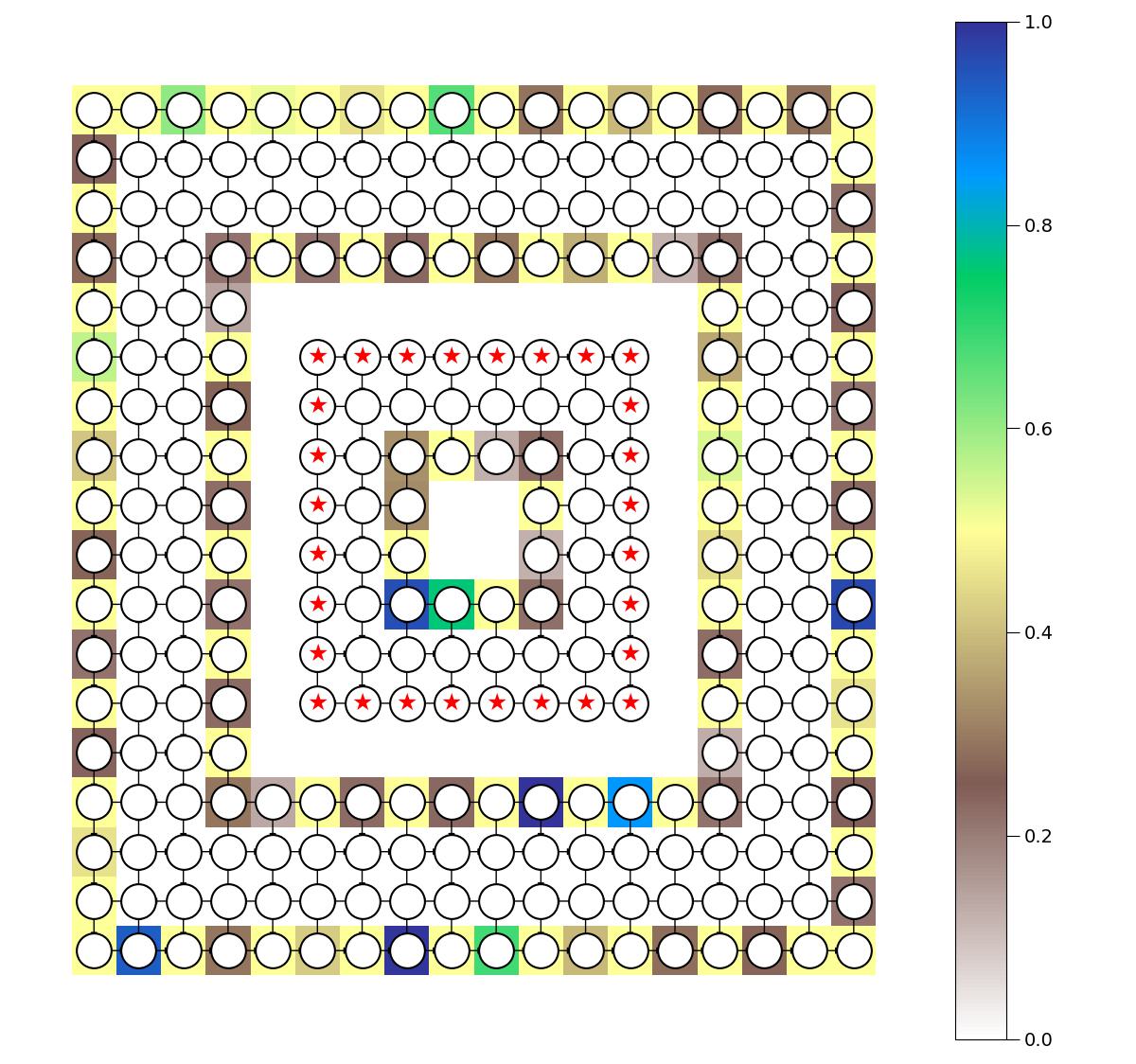}}
	\caption{{\bf Lattice with selectively removed units to restrict oscillations to desired bands}: Units have been removed from specific parts of the lattice, shown in white, creating defects in both the center and surrounding area (white frame). The colored background highlights regions where the remaining units  oscillate with nonzero frequency. In contrast, units without colored background settle into OD-states with zero frequency. Note that units at the outer edge of the white frame oscillate, those at the inner edge (marked with asterisks) are in OD-states in spite of their location at an edge.  The phase color coding of the nodes is omitted in this figure for clarity. The parameters are identical to those in Fig.\ \eqref{Picture8}. For further explanations see the main text.
	}
	\label{Picture110}
\end{figure}

In Fig.\ \eqref{Picture110} we see three layers with colored background, corresponding to edge oscillations, a white square in the center of the figure as well as  a white frame where SL-units are removed. Note that the white frame of defects is surrounded by oscillating units only on one side; on the other side of the white frame, the background of the SL-nodes is white, meaning that these units are in an OD-state. Being located on the edge is not sufficient for a unit to oscillate. This means if one wants to design oscillating  layers  in a bulk of OD-states, ``frozen'' to fixed-points, the choice of blocking areas where units are removed should respect the criteria as derived from the weak-coupling limit: the links orthogonal to the edges which connect the respective bulk and the edges should have weak couplings assigned. This condition is violated for the ``frozen'' edge (with white background color and red asterisks).

\subsubsection{Bidirectional coupling}\label{subsub344}
In analogy to the SSH-model, we have focussed so far on directed couplings. Following the indicated directions of couplings along the edge of the grid, it may suggest a directed transport (clockwise) and opposite to the anti-clockwise ``chirality'' of unit cells. We have not observed any indications for a wave propagating along the edge in a clockwise direction. Therefore the directed version of the couplings seems to be irrelevant for the phenomenon of edge oscillations. However, if we compare the dynamics of a pair of SL-units, coupled in a directed way or in both directions on the edge according to
\begin{equation}
	\mathcal{A}_{ji} =
	\begin{cases}
		-w', & \text{if } \mathcal{A}_{ij} = -w, \\
		-s', & \text{if } \mathcal{A}_{ij} = -s,
	\end{cases}
\end{equation}
where $w' \in [0,w]$ and $s' \in [0,s]$,
we see from Fig.\ \eqref{fig201202} that directed coupling between them facilitates oscillations: We compare two strongly coupled units along the link between two sites on the edge  with directed coupling in (a), and bidirectional coupling of equal strength in (b) for otherwise the same parameters.
\begin{figure}[htp]
\centerline{\includegraphics[width=0.46\textwidth]{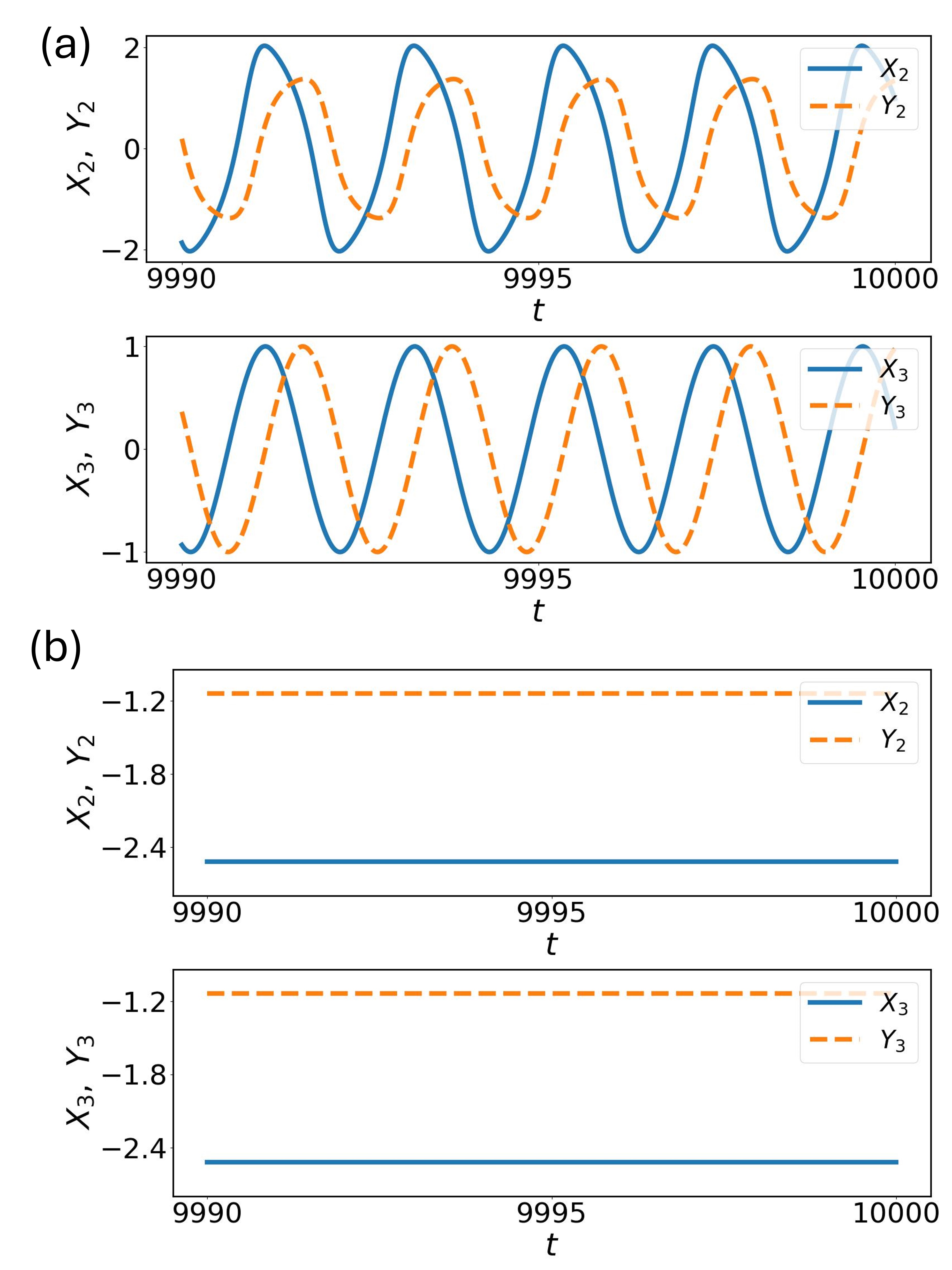}}
\caption{{\bf Directed coupling supports oscillations }:  (a) Oscillating pair of units at two sites on the edge connected by directed coupling of strength $s=8$, when such oscillating pairs are isolated in the weak-coupling limit ($w=0$). (b) Fixed-point values of units at rest at the same sites along the edge when bidirectionally coupled with the same strength.  All other parameters are kept fixed as in Fig.\ \eqref{Picture3}, but with $w=0$.
	}
\label{fig201202}
\end{figure}

If we choose bidirectional coupling on the entire grid with initial conditions as in Fig.\ \eqref{Picture3} and  $w' = w = 0.001$ and $s' = s = 8$ on a $10 \times 10$ lattice of SL oscillators, the corresponding dynamics are illustrated in Fig.\ \eqref{Picture16}, where only the four corner oscillators exhibit nonzero normalized frequencies, indicating that they are oscillating. These corner oscillators are incoherent and oscillate within the range $[-1,1]$, while the rest of the oscillators exhibit zero normalized frequencies, signifying their convergence to OD-states.  
The remaining edge oscillators settle into two distinct OD-states that are symmetric around the origin, depending on the initial conditions. Similarly, all bulk oscillators converge into two symmetric OD-states with coordinates  larger than those of the units on the edge. This means that for bidirectional coupling of equal strength we do not observe edge-localized oscillations.

\begin{figure}[htp]
	\centerline{\includegraphics[width=0.4\textwidth]{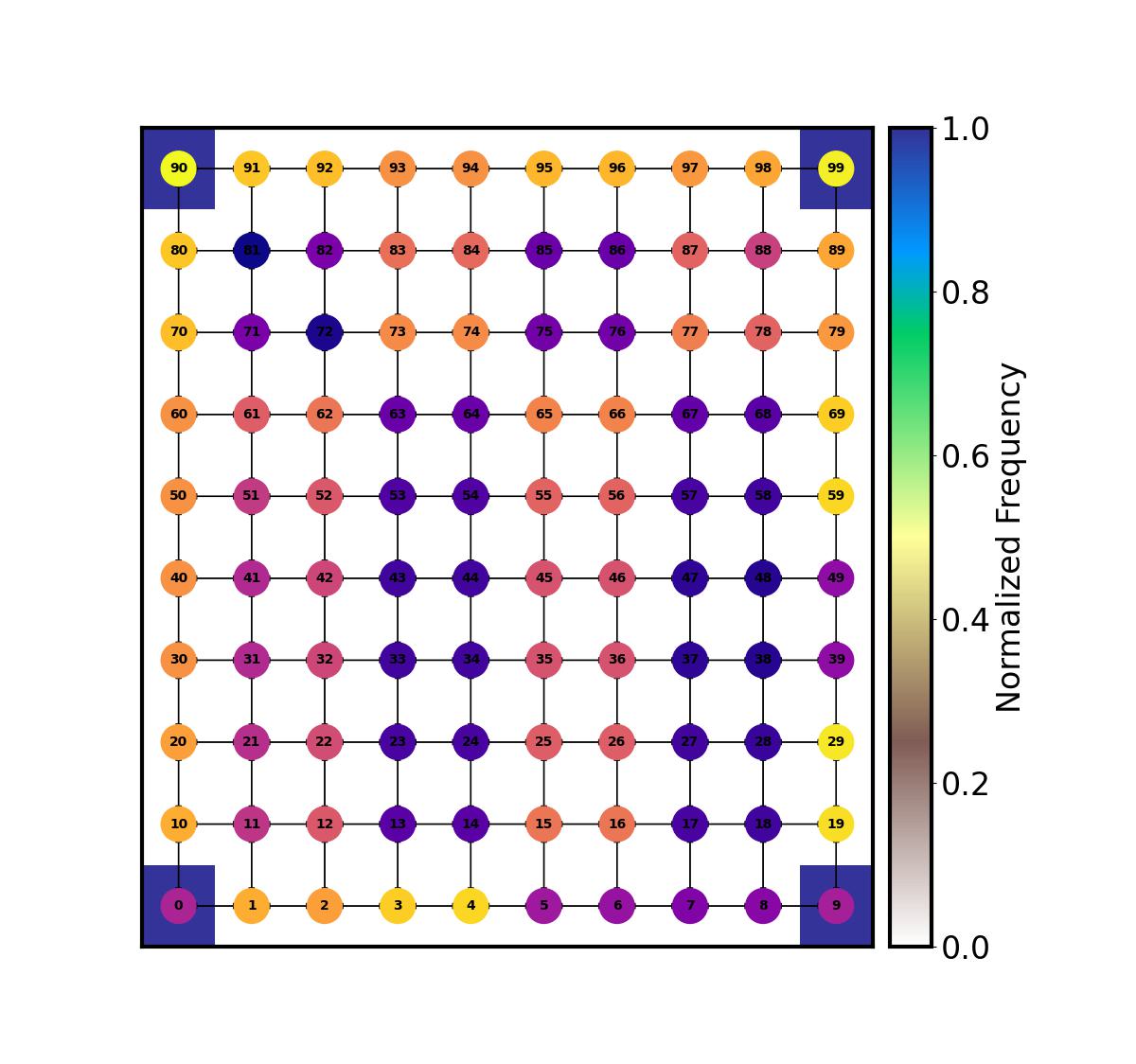}}
	\caption{{\bf Impact of bidirectional coupling on edge oscillations}: Snapshot of the dynamics on a \(10 \times 10\) lattice of SL-oscillators under bidirectional coupling with \( w' = w = 0.001 \) and \( s' = s = 8 \). Only the four corner oscillators exhibit nonzero normalized frequencies, oscillating within the range $[-1,1]$, while all other oscillators converge to OD-states. Color coding as in Fig.\ \eqref{Picture7}. All other parameters are kept fixed as in Fig.\ \eqref{Picture3}.
	}
	\label{Picture16}
\end{figure}

Numerical analysis reveals that for $w' = 57\%$ of $w$ and $s' = 57\%$ of $s$ (with $w=0.001$ and $s=8$), the bidirectionally coupled system continues to exhibit edge oscillations while bulk oscillators remain in OD-states. Once $w' = 58\%$ of $w$ and $s' = 58\%$ of $s$, the system transitions to a state similar to Fig.\ \eqref{Picture16}. It should be noted that edge oscillators in this case settle into different OD-states than the bulk clusters.  

This critical percentage threshold changes slightly when self-feedback is introduced. For $g=1$, edge oscillations persist up to $w' = 76\%$ of $w$ and $s' = 76\%$ of $s$ with $w=0.001$ and $s=8$. However, at $w' = 77\%$ of $w$ and $s' = 77\%$ of $s$, the system again exhibits behavior similar to Fig.\ \eqref{Picture16}.

\par If we completely reverse the directionality by setting \( w' = 0.001 \) and \( s' = 8 \) with \( w = s = 0 \) while maintaining the other parameters as in Fig.\ \eqref{Picture3}, such a reversal does not qualitatively alter the results of edge oscillations combined with bulk oscillators in OD-states.

\section{Topological origin of edge-localized oscillations}\label{sec4}

\subsection{Derivation of an effective Hamiltonian}\label{sec41}
We derive an analogue of the Bloch-Hamiltonian to analyze the bulk band structure in view of topological invariants, which are supposed to explain the emergence of edge states in terms of  non-vanishing Zak-phases. We start from Eq.\ \eqref{eq:coupled_system}. Linearizing Eq.\ \eqref{eq:coupled_system} around the stationary solutions $(X_0,Y_0)$ according to $X_i=X_{i0}+\delta X_i$ and $Y_i=Y_{i0}+\delta Y_i$ leads to an equation of the form
\begin{equation}\label{eq2}
\frac{d}{dt}\begin{pmatrix} \delta X_i\\ 
\delta Y_i\end{pmatrix}=\begin{pmatrix}J_{XX}\delta_{ij}-\sum_j A_{ij} &J_{XY}\delta_{ij}\\ J_{YX}\delta_{ij}& J_{YY}\delta_{ij}\end{pmatrix}\begin{pmatrix}\delta X_j\\ \delta Y_j\end{pmatrix},
\end{equation}
with $i\in\{1,...,L_x\times L_y\}$ for a grid of size $L_x\times L_y$ and $J_{XX},...,J_{YY}$ the respective components of the Jacobian, evaluated at the stationary point. This equation is of first order in time and reminds to the Schr\"odinger equation, therefore, the matrix on the right hand side will later be termed Hamiltonian. Since we are interested in stationary solutions, we separate variables and have to solve an eigenvalue equation that reads in coordinate space $E \psi_i=\Omega_{ij} \psi_j$ with $\psi_j$ the space-dependent part of the solution $(\delta X_j,\delta Y_j)^T$ and $\Omega_{ij}$ the matrix as in Eq.\ \eqref{eq2},  $E$ is an eigenvalue. 

Since we want to distinguish between internal and external couplings and later  perform the Fourier transform only with respect to external coordinates, we split an index $i$ as $(x,y)_I$ with $x\in \{1,...,N_x\}$, $y\in \{1,...,N_y\}$ indicating the external coordinate of a cell, while $I\in\{A,B,C,D\}$, labelling the internal coordinate within a cell. In this notation, Eq.\ \eqref{eq:coupled_system} reads
\begin{eqnarray}\label{eq1:coupled_system}
\dot{X}(x,y)_I &=& F(X(x,y)_I,Y(x,y)_I)\nonumber \\
 &-& \sum_{x^\prime,y^\prime;I^\prime} \mathcal{A}((x,y)_I,(x^\prime,y^\prime)_{I^\prime}) X(x^\prime,y^\prime)_{I^\prime}, \nonumber \\
\dot{Y}(x,y)_I &=& G(X(x,y)_I,Y(x,y)_I),
\end{eqnarray}
with 	$x^\prime\in \{1,...,N_x\}$, $y^\prime\in \{1,...,N_y\}$, and $I^\prime\in\{A,B,C,D\}$.
The on-site dynamics on all grid points is in principle the same, also the type of directed coupling in and outside of the cell. Although the splitting into internal and external coordinates  looks artificial at this point,  the very ratio of both coupling strengths turns out to be essential for the existence of edge states.

Next we evaluate the interaction term at given coordinate $(x,y)_I$  for directed interactions as in Fig.\ \eqref{Picture1}. It is given by
\begin{eqnarray}\label{eqfou}
&-&\sum_{x^\prime,y^\prime, I^\prime} A((x,y)_I, (x^\prime,y^\prime)_{I^\prime})\delta X((x^\prime,y^\prime)_{I^\prime})\\ \nonumber
&=&-\begin{pmatrix}\gamma_{in} \delta X(x,y)_D +\gamma_{ex}\delta X(x,y-1)_B\\ 
\gamma_{in} \delta X(x,y)_A +\gamma_{ex}\delta X(x-1,y)_C  \\ 
\gamma_{in} \delta X(x,y)_B +\gamma_{ex}\delta X(x,y+1)_D  \\ 
\gamma_{in} \delta X(x,y)_C +\gamma_{ex}\delta X(x+1,y)_A  
\end{pmatrix}
\end{eqnarray}
\\
For periodic boundary conditions, we next perform a Fourier transform only with respect to the external coordinates to reduce the eigenvalue problem to that of a single cell, using
\begin{equation}
\begin{pmatrix} \delta X(x,y)_I\\ 
\delta Y(x,y)_I\end{pmatrix}=\frac{1}{\sqrt{N_x N_y}}\sum_{k_x,k_y} e^{-i(k_xx+k_yy)}\begin{pmatrix}\delta \tilde{X}(k_x,k_y)_I\\ 
\delta \tilde{Y}(k_x,k_y)_I\end{pmatrix}
\end{equation}
with $(k_x,k_y)$ running over the first Brioullin zone, to obtain with $\begin{pmatrix}\delta \tilde{X}(k_x,k_y)_I\\ 
\delta \tilde{Y}(k_x,k_y)_I\end{pmatrix} \equiv \delta\tilde{\psi}_{I^\prime}(k_x,k_y)$ the eigenvalue equation
\begin{equation}
E \delta\tilde{\psi}_{I}(k_x,k_y) = \left(J\delta_{II^\prime}+\tilde{A}_{II^\prime}\small{\begin{pmatrix}1&0\\0&0\end{pmatrix}}\right)
\delta\tilde{\psi}_{I^\prime}(k_x,k_y),
\end{equation}
where $J\delta_{II^\prime}$ is the $2\times 2$-Jacobian at $I^\prime=I$, with $I\in\{A,B,C,D\}$, while  $\tilde{A}_{II^\prime}$ results from the Fourier transform of Eq.\ \eqref{eqfou}, both combined to the following $8\times 8$-matrix:

\begin{equation}\label{ourhamiltonian}
	\mathcal{H}(k_x, k_y) = 
	\scalebox{0.8}{$
		\begin{bmatrix}
			J_{XX}^{A} & J_{XY}^{A} & w e^{i k_y} & 0 & 0 & 0 & s & 0 \\
			J_{YX}^{A} & J_{YY}^{A} & 0 & 0 & 0 & 0 & 0 & 0 \\
			s & 0 & J_{XX}^{B} & J_{XY}^{B} & w e^{i k_x} & 0 & 0 & 0 \\
			0 & 0 & J_{YX}^{B} & J_{YY}^{B} & 0 & 0 & 0 & 0 \\
			0 & 0 & s & 0 & J_{XX}^{C} & J_{XY}^{C} & w e^{-i k_y} & 0 \\
			0 & 0 & 0 & 0 & J_{YX}^{C} & J_{YY}^{C} & 0 & 0 \\
			w e^{-i k_x} & 0 & 0 & 0 & s & 0 & J_{XX}^{D} & J_{XY}^{D} \\
			0 & 0 & 0 & 0 & 0 & 0 & J_{YX}^{D} & J_{YY}^{D}
		\end{bmatrix}
		$}
\end{equation}

With the Jacobians evaluated at the stationary state components at sites A,B,C,D for the three types of on-site dynamics, this is our analogue of the Bloch Hamiltonian, from which we want to determine the bulk bands and the Zak phases.
	
\subsection{Symmetries of the Hamiltonian}\label{sec42}
\par {\bf Non-Hermiticity.} A first observation is that our Hamiltonian is non-Hermitian because of $\mathcal{H}\neq\mathcal{H}^\dag$. Physical reasons for the non-Hermiticity in our case are the directed couplings that are responsible for the off-diagonal components, as well as the on-site dynamics, leading to the block-diagonal terms which  formally play an analogous role to gain and loss terms in Bloch Hamiltonians for extended non-Hermitian SSH-models \cite{lieu2018topological, lai2025non}. Therefore we expect complex eigenvalues, the need to distinguish right and left eigenvectors and possible exceptional points \cite{heiss2012physics,miri2019exceptional,panahi2024higher}, at which two or more eigenvalues and corresponding eigenvectors coalesce. 
Actually, we do find exceptional points in the spectrum of SL-oscillators for $s=w=8.0$ (other parameters: $\alpha_i=\alpha=1.0; \omega_i=\omega=3.0, \beta_i=1,g=0$), a parameter set  that lies deeply in the region of OD-states with no significant change in the sensitivity of the OD-states close to the exceptional points. Neither did we find exceptional points for other on-site dynamics in a parameter region of interest, where edge-localized oscillations occur. Therefore, we will not discuss exceptional points any further.

Moreover, as indicated before, for non-Hermitian systems, in general the bulk-boundary correspondence need no longer hold, due to which the topology of the bulk bands explains the existence of edge states located at the boundaries. As we shall see, our bulk bands show non-trivial topology. Of particular interest in relation to topology are symmetries of the Hamiltonian, here these are time reversal and inversion symmetry that we discuss next. \\

\par {\bf Inversion symmetry.} Inversion symmetry is expressed as  
\begin{equation} \label{inversion}  
	\mathcal{I} \mathcal{H}(k_x, k_y) \mathcal{I}^{-1} = \mathcal{H}(-k_x, -k_y),  
\end{equation}  
where \(\mathcal{I} = \mathbf{\sigma}_x \otimes I_4\), with \(\mathbf{\sigma}_x = \begin{bmatrix}  
	0 & 1  \\  
	1 & 0  
\end{bmatrix}\) being the Pauli matrix and \(I_n\) denoting the \(n\)-th order identity matrix. \(\mathcal{I}\) is a square, real, involutory matrix and unitary. This symmetry imposes a quantization condition on the integrated Berry connection in the Hermitian limit \cite{zak1989berry}. Our effective Hamiltonian $\mathcal{H}(k_x, k_y)$ satisfies inversion symmetry under the following conditions. Suppose, the four strongly unidirectionally coupled oscillators with coupling strength $s$ converge to the stationary states $(X_A, Y_A), (X_B, Y_B), (X_C, Y_C), (X_D, Y_D)$, respectively. If $(X_A, Y_A) = (X_C, Y_C)$ and $(X_B, Y_B) = (X_D, Y_D)$, then $\mathcal{H}(k_x, k_y)$ obeys inversion symmetry.

\par Note that there may be additional conditions under which the inversion symmetry \eqref{inversion} holds. For instance, in the cases of the Stuart-Landau and brusselator models, this symmetry is preserved under specific constraints such as $(X_A, Y_A) = (-X_C, -Y_C)$ and $(-X_B, -Y_B) = (X_D, Y_D)$. However, our numerical simulations did not find the steady states of four strongly coupled oscillators satisfying these conditions. In fact, for coupled brusselators, such symmetry is inherently not allowed, as the amplitude of the state variables must remain non-negative. Moreover, this symmetry ensures that the matrices $\mathcal{H}(k_x, k_y)$ and $\mathcal{H}(-k_x, -k_y)$ have the same eigenvalues, as reflected in the band structure, but not necessarily the same eigenvectors. This is easily proven (see lemma \eqref{Theo1}  in Appendix \eqref{appendixB}). 
 
{\bf Time reversal symmetry.} Our Hamiltonian \(\mathcal{H}(k_x, k_y)\) satisfies also time-reversal symmetry, which is expressed as:
\begin{equation} \label{time-reversal}  
	\mathcal{H}(k_x, k_y) = \mathcal{H}^{*}(-k_x, -k_y),  
\end{equation}
where \(\mathcal{H}^{*}(k_x, k_y)\) denotes the complex conjugate of \(\mathcal{H}(k_x, k_y)\). This symmetry implies that the eigenvalues of \(\mathcal{H}(k_x, k_y)\) and \(\mathcal{H}(-k_x, -k_y)\) appear in complex conjugate pairs. 
Also this statement is easily proven  (see lemma \eqref{Theo2} in Appendix \eqref{appendixB}).

{\bf Other symmetries.} We want to mention also two other symmetries that are  often discussed in the context of non-Hermitian Hamiltonians, which our Hamiltonian violates. The first one is 
PT symmetry, which is expressed as  
\begin{equation} \label{PT}  
	(\mathcal{I}\mathcal{T}) \mathcal{H}(k_x, k_y) (\mathcal{I}\mathcal{T})^{-1} = \mathcal{H}(k_x, k_y).  
\end{equation}  
Here, the time-reversal operator \(\mathcal{T}\) acts as the identity matrix \(I_8\). A more detailed explanation of why our effective Hamiltonian, which contains complex terms, does not satisfy this symmetry is provided  in lemma \eqref{Theo4} of Appendix~\ref{appendixB}. 

In relation to possible chiral currents, the sublattice symmetry is of interest, given as
\begin{equation} \label{Sublattice}  
	\mathcal{S}\mathcal{W}(k_x, k_y)\mathcal{S}^{-1} = -\mathcal{W}(k_x, k_y).  
\end{equation}  
Here, the symmetry operator is defined as \(\mathcal{S} = I_4 \otimes \sigma_z\). To see why our Hamiltonian violates this symmetry we remove all Jacobian terms from the principal block diagonal of \(\mathcal{H}(k_x, k_y)\). The resulting matrix \(\mathcal{W}(k_x, k_y)\) then satisfies the relation  
\begin{equation}  
	\mathcal{S}\mathcal{W}(k_x, k_y)\mathcal{S}^{-1} = \mathcal{W}(k_x, k_y),  
\end{equation}  
which contradicts the sublattice symmetry condition Eq.\ \eqref{Sublattice}.

\subsection{Zak phases and dispersion bands}\label{sec44}
Since we implement the very same geometry as the one of the two-dimensional SSH-model on a square lattice, we make also use of a topological  quantity that characterizes the non-trivial topology in the SSH-model. This is the Zak-phase \cite{zak1989berry, zheng2024topological, liu2017novel}, which is a geometric phase. In the SSH-model it is the phase acquired by a Bloch wavefunction as it evolves across the Brillouin zone. In two dimensions, we have to deal with a vector of Zak phases. There, for the j-th energy band the vectored Zak phase $(\Phi_x^{(j)},\Phi_y^{(j)})$ is defined as the line integration of the Berry connection (whose curl yields the Berry curvature). 
The Berry connection of the j-th energy band is defined as $A_j=(a_j(k_x),a_j(k_y))^T$, where $a_j(k_m)=i(\psi^{(j)}(k_x,k_y)^{(L)\dag} \frac{\partial}{\partial k_m}\psi^{(j)}(k_x,k_y)^{(R)})$, ($m\in\{x,y\}$) has the typical form of a geometric connection and $\psi^{(j)}(k_x,k_y)^{(L),(R)}\equiv \psi^{(j)}(\mathbf{k})^{(L),(R)}$ are a left (L) or right (R) eigenvector of the analogue Eq.\ \eqref{ourhamiltonian} of the reduced Bloch-Hamiltonian in momentum space, which is non-Hermitian. Thus, the definition is \cite{obana2019topological, nelson2024nonreciprocity}
\begin{eqnarray}\label{eqzak}
\Phi_x^{(j)}&=&\dfrac{i}{2\pi}\;\int_{-\pi}^{+\pi}dk_x \psi^{(j)}(\mathbf{k})^{(L)\dag} \partial_{ k_x}\psi^{(j)}(\mathbf{k})^{(R)} \nonumber \\
\Phi_y^{(j)}&=&\dfrac{i}{2\pi}\;\int_{-\pi}^{+\pi}dk_y \psi^{(j)}(\mathbf{k})^{(L)\dag} \partial_{ k_y}\psi^{(j)}(\mathbf{k})^{(R)} 
\end{eqnarray}

Before Computing the Zak Phase components, we normalize the right and left eigenvectors at each $\mathbf{k}$-point to unit norm.

A nonzero vector of Zak phases ($\vec{\Phi}^{(j)} \neq 0$) indicates a nontrivial topological phase, which often corresponds to the presence of edge states in a finite geometry.
When it vanishes, it typically corresponds to a trivial insulator in the context of condensed matter physics without protected edge states.	

We numerically compute the Zak phase along both \( k_x \) and \( k_y \) directions. The momentum components are discretized into $101$ points each. The system parameters are set as: $\alpha_i = \alpha = 1.0, \quad \omega_i = \omega = 3.0, \quad s = 8.0, \quad w = 0.001, \quad \beta_i=1, \quad g=0.$ For our choice of Hamiltonian Eq.\ \eqref{ourhamiltonian} and for this particular set of parameter values, the computed 2D Zak phases take the values $(\Phi_x^{(j)},\Phi_y^{(j)})= (0, 0)$, 
for each of the eight bands. More precisely, for a grid size of $101$ points, the imaginary part is of the order of $1.5 \times 10^{-10}$, for $201$ points, it reduces to $7.6 \times 10^{-11}$, and for $301$ points, it further decreases to $5.1 \times 10^{-11}$.	
When the coupling values of \( w \) and \( s \) are interchanged, (corresponding to a switch in the geometries in analogy to Fig.~\ref{Picture2} (a) and (b)), i.e., setting \( w = 8 \) and \( s = 0.001 \), the computed 2D Zak phases change to $(\Phi_x^{(j)},\Phi_y^{(j)})= (\pi,\pi)$ for each of the eight bands. Here, more precisely, for a grid size of $101$ points the imaginary part is of the order of $~0.0156$, for $201$ points, it reduces to $~0.0078$, and for $301$ points, it further decreases to $~0.0052$.\\


Had we calculated the Hamiltonian for a unit cell with weak internal coupling assignment (given explicitly in Eq.\ \eqref{ouroldhamiltonian}), the choice of $s=8.0$ and $w=0.001$ would have led to bulk bands with Zak phases $(\pi,\pi)$, whereas $s=0.001$ and $w=8.0$ would have led to vanishing Zak phases.

How are these topological features of the bulk bands reflected in edge-localized oscillations or  everywhere OD-states on a finite grid with alternating coupling geometry and different sizes? The answer depends on the coupling strength with which  boundary nodes of the bulk are connected to edge nodes. The corresponding links, connecting these nodes,  are orthogonal to the edge. In general, if these couplings are weak, we observe edge-localized oscillations, 
if they are strong, 
we observe OD-states everywhere. In Appendix \eqref{app5} we discuss in some detail the correspondence between the size of the lattice geometry, leading to edge-localized oscillations, and the choice of Hamiltonian (unit cell), leading to non-vanishing Zak phases of the bulk bands.\\
It is this correspondence that we take as a sign for topological protection of our edge oscillations. The computed 2D Zak phases remain nonzero and equal to $(\pi, \pi)$, when we activate the linear self-coupling by setting \( g = 1 \) and \( g = 2 \), while keeping all other parameter values the same as earlier. Even when we compute the 2D Zak phases for the discussed result with repulsive coupling of the same strength in coupled SL-oscillators, we find that the Zak phases remain nonzero.\\
System phases, characterized  by  vanishing Zak phases, are termed trivial in the sense that topological protection is lost. Accordingly, our numerical simulations reveal that in this regime, all coupled units transition to two clusters of OD-states, where the cluster composition depends on the initial conditions.\\

{\bf Dispersion bands.} The underlying dispersion bands to the different values of the Zak phases are shown in Figs.\ \eqref{band2d} and \eqref{fig102104} (a)-(d). For each momentum point \( (k_x, k_y) \) sampled on a uniform \( 101 \times 101 \) grid in the Brillouin zone, we compute the eigenspectra of both versions of the Hamiltonian, evaluated at \( s = 8.0 \), \( w = 0.001 \) (version A) or at \( s = 0.001 \), \( w = 8.0 \) (version B), respectively. This yields sorted eigenvalues and their associated right eigenvectors. The spectral norm difference between the sorted eigenvalues of version A and version B 
reveals that the real parts of the eigenvalues are identical up to numerical precision. However, the imaginary parts differ significantly in bands \( 2 \), \( 3 \), \( 6 \), and \( 7 \) for a subset of momentum points. 

In Figs.\ \eqref{band2d} (a)-(d) we show the dispersion bands in two  dimensions for Zak phases $(0,0)$ (upper two panels, for version A) and Zak phases $(\pi,\pi)$ (lower two panels, for version B). The two-dimensional plots of the real and imaginary parts of the eight eigenvalues as function of   $k_x$ and fixed $k_y \approx 0$ give a better resolution than the three-dimensional bands, plotted over the first Brioullin zone. 

\begin{figure}[htp]
\centerline{\includegraphics[width=0.46\textwidth]{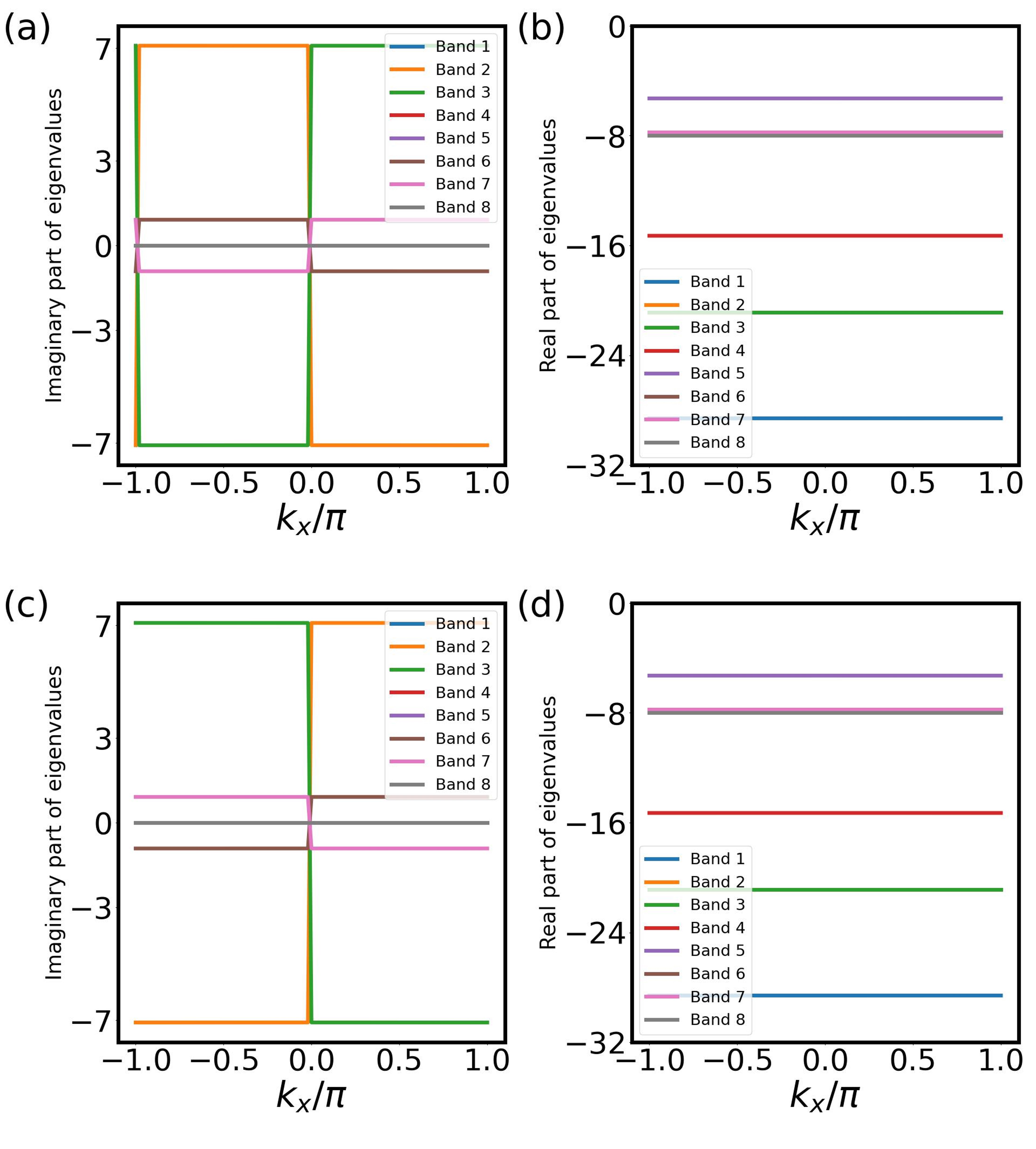}}
\caption{{\bf Eigenvalues for the Hamiltonian $\mathcal{H}$ with varying $k_x$ and fixed $k_y \approx 0$}:  (a) and (c): Imaginary part of the eigenvalues, showing a switch between the bands corresponding to the nonzero imaginary eigenvalues from positive to negative (or vice versa). (b) and (d): Real part of the eigenvalues, revealing six distinct bands, with four bands being completely real and two real parts arising from the complex conjugate pairs of nonzero imaginary eigenvalues.  Panels (a) and (b): Zak phases $(0,0)$, panels (c) and (d): Zak phases $(\pi,\pi)$. Other parameters $\alpha_i=\alpha =1.0, \omega_i=\omega=3.0$, $s=0.001$ and $w=8.0$ (c) and (d) and vice versa (a) and (b).}
\label{band2d}
\end{figure}
\vskip5pt

In Figs.\ \eqref{fig102104}(a)-(d) we show the corresponding dispersion bands in three dimensions for Zak phases $(0,0)$ ((a) and (b)) and Zak phases $(\pi,\pi)$ ((c) and (d)). The figure also reveals that the imaginary parts for bands 2, 3,  6, and 7 differ. The symmetry of the eigenvalues around the origin is a direct consequence of inversion and time-reversal symmetries.
\begin{figure*}[htp]
\includegraphics[width=0.75\textwidth]{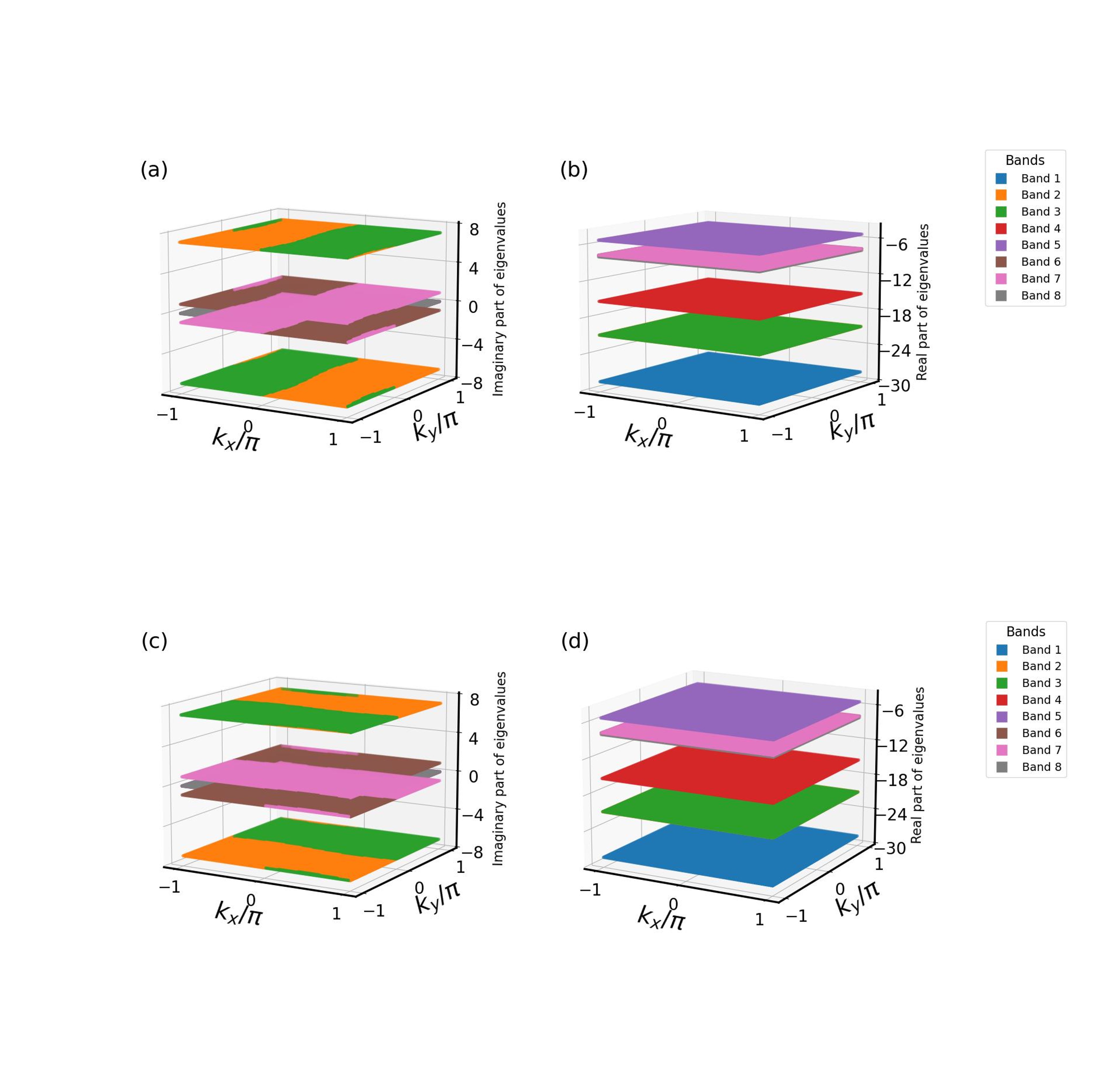}
	\caption{{\bf Dispersion relations for the Hamiltonian $\mathcal{H}$ with varying $k_x$ and  $k_y $}:  (a) and (c): Imaginary part of the eigenvalues, showing a switch between the bands (green and orange), corresponding to the nonzero imaginary eigenvalues switching from positive to negative (or vice versa). (b) and (d): Real part of the eigenvalues, revealing six distinct bands, with four bands being real and two arising from the complex conjugate pairs of nonzero imaginary eigenvalues. Panels (a) and (b): Zak phases $(0,0)$, panels (c) and (d): Zak phases $(\pi,\pi)$. Other parameters as in Fig.\ \eqref{band2d}.
}
\label{fig102104}
\end{figure*}

\begin{figure*}[htp]
\includegraphics[width=1.0\textwidth]{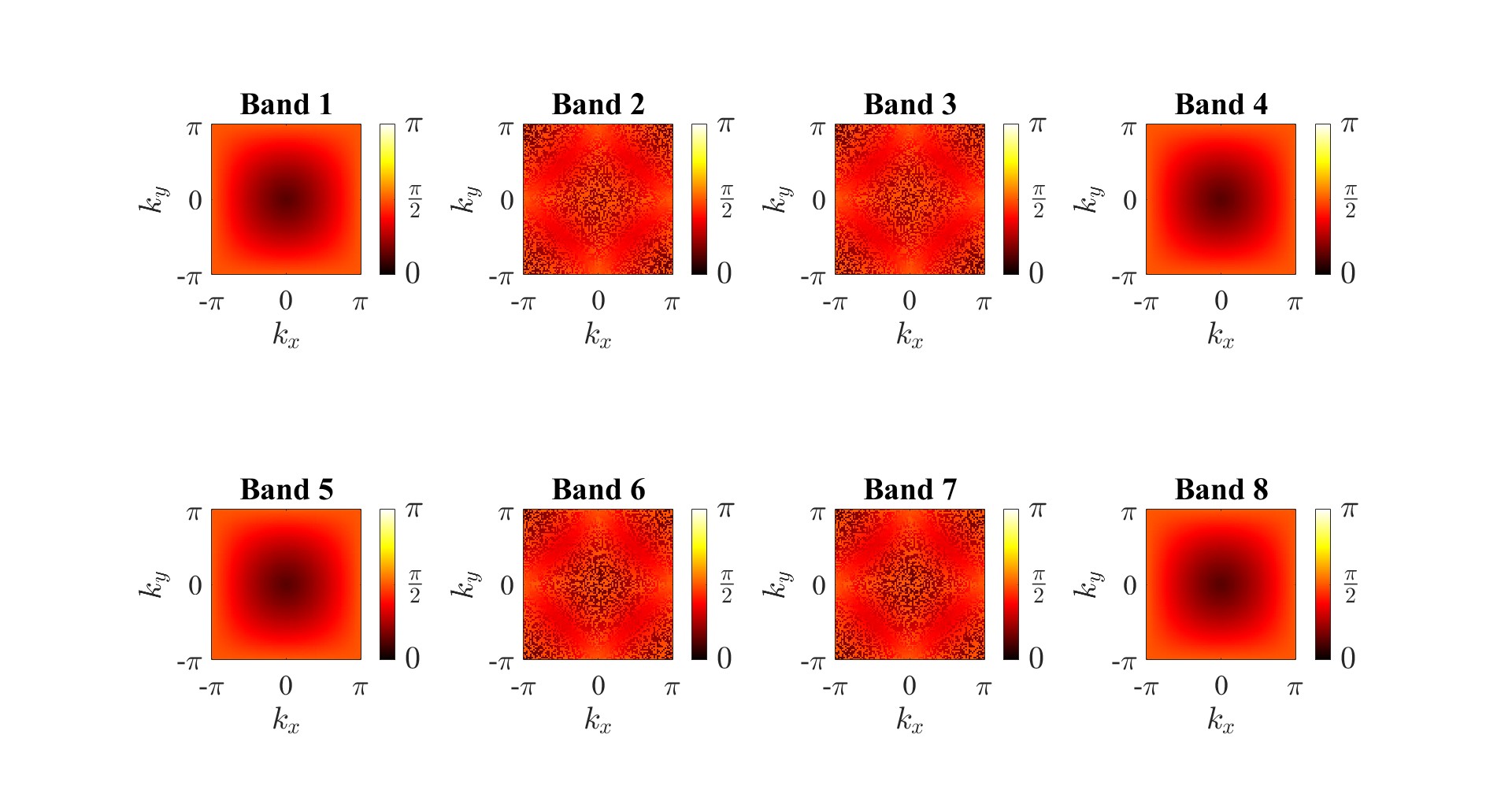}
\caption{ Different Zak phases caused by the parallel transport of different right eigenvectors of the Hamiltonian, monitored by their angle differences \( \theta_b^{(R)}(k_x, k_y) \) when evaluated for two parameter regimes: case A with \( s = 8 \), \( w = 0.001 \) (Zak phases (\( 0 \,0\)), and case B with \( s = 0.001 \), \( w = 8 \) (Zak phases $(\pi, \pi)$). Note that the eigenvectors differ while the eigenvalues for bands \( 1 \), \( 4 \), \( 5 \), and \( 8 \) are identical.
	Each subplot corresponds to a band index \( b = 1, \ldots, 8 \) and shows a heatmap over the Brillouin zone, sampled on a \( 101 \times 101 \) grid. The color scale highlights the angular deviation $\theta_b^{(R)}(k_x, k_y)$ ranging from \( 0 \) (perfectly aligned) to \( \pi \) (anti-parallel). Angle differences near \( \pi/2 \) are present, which indicates that the corresponding right eigenvectors are nearly orthogonal at certain momentum points \( (k_x, k_y) \). 
	}\label{figangle}
\end{figure*}

It should be noticed that the difference in the topological characteristic is observed for all of the eight bands from version A and B of the effective Hamiltonian, even for identical bands, such as bands $1,\;4,\;5,\;8$. This must result from differences in the eigenvectors of the effective Hamiltonian while parallel transported via the Berry connection in the Brioullin zone for $k_x$ and $k_y$ from $-\pi$ to $\pi$. Therefore, we consider the angle between the right eigenvectors and define:
\begin{equation}
\theta_b^{(R)}(k_x, k_y) = \cos^{-1}\left( \left| \frac{\langle u_b, v_b \rangle}{\|u_b\| \cdot \|v_b\|} \right| \right),
\end{equation}
with normalized eigenvectors \( u_b \) and \( v_b \).
Each subplot in the corresponding figure represents one of the 8 bands and displays a heatmap of \( \theta_b^{(R)}(k_x, k_y) \) over the Brillouin zone. The color scale spans from \( 0 \) (perfect alignment) to \( \pi \) (anti-parallel), providing a detailed visualization of how the eigenvectors differ across momentum space.
\vskip5pt
In summary, the values of the Zak phases, being $\pi$ or $0$ for the bulk bands, distinguish between the system's phases, characterized by different stationary states: these are (i) edge oscillations combined with OD-states in the bulk in the non-trivial system's phase, where the Zak phases are $\pi$ for both components, or, (ii) no edge oscillations and everywhere OD-states, where the Zak phases are both $0$. We term it also a trivial phase when we observe everywhere oscillations, although in this case, the Zak phase  cannot be determined from the eigenvectors of the Hamiltonian, as the Hamiltonian is derived from a linear stability analysis around the fixed-points.

As the Zak phases are topological invariants, the nonzero value of the Zak phases justifies to call the edge oscillations topologically protected. As we have seen, the protection is manifest in the robustness against strong parameter mismatch, noise, and deformations of the grid. Note that the non-vanishing Zak phases of the non-trivial system's phase (characterized by stationary states with edge oscillations for open boundary conditions) have been  determined for the bulk bands, obtained for periodic boundary conditions in both directions. This fact amounts to a realization of the bulk-boundary correspondence. The fact that the values are $0$ or $\pi$ results from the inversion symmetry of the effective Hamiltonian. In general, in $1$D Hermitian systems with inversion symmetry, the Zak phase of an isolated band is quantized to $0$ or $\pi$~\cite{jiao2021experimentally, xiao2016coexistence}; this quantization is related to the  symmetries of the Bloch wave at the center and the boundary of the Brioullin zone. In our derivation,  an ansatz in terms of ``Bloch waves'' is included in the discrete Fourier transform with respect to the external coordinates. The visible  angle difference between the two eigenvectors at the center and the boundary of the Brioullin zone  in Fig.\ \eqref{figangle} does not prove but is compatible with different symmetries of the eigenvectors at the center and the boundaries for case A and B, respectively.

\section{Other on-site nonlinearities}\label{sec5}
The main qualitative features that we observed for coupled SL-oscillators remain valid when the on-site dynamics is replaced by activator-repressor or brusselator systems. Oscillations, restricted to the edge of the grid, are protected by non-vanishing Zak phases of the bulk bands, while  vanishing Zak phases correspond to OD-states everywhere on the grid. Differences occur in the quantitative sensitivity to perturbations like noise, and in the substructure of edge oscillations (frequency chimera or not), or OD-states (structure of coexisting different OD-states).

\subsection{Activator-repressor on-site dynamics}\label{subsec351}
For the on-site dynamics of Eq.~(\ref{BFU}) we summarize numerical results for different parameter choices in Table~\ref{tab:oscillator_behavior}. Given that we already have three primary coupling parameters \( s \), \( w \), and \( g_i = g \), we fix other parameters in Eq.~(\ref{BFU}). Specifically, we set \( b = 0.01 \), \( \gamma = 0.01 \), and \( K = 0.02 \) throughout our study. 
To ensure a diverse set of initial conditions, we randomly initialize the system states within the domain \([0,1] \times [0,1]\) for all simulations.
	\renewcommand{\arraystretch}{1.3}
	\begin{table}[h]
		\centering
		\setlength{\tabcolsep}{6pt}
		\begin{tabular}{|c|c|c|c|l|l|}
			\hline
			$\alpha$ & $g$ & $w$ & $s$ & \multicolumn{1}{c|} {\textbf{bulk units}} & {\textbf{edge units}} \\ 
			\hline
			
	80 & 0.0 & 0.0001 & 0.5 &  OD, coh &  OS, incoh\\
			\hline
	80 & 1.0 & 0.0003 & 1.5 & OD,  coh & OS, incoh	\\
			\hline		
		60 & 0.0 & 0.0001 & 0.5 &  OD,  coh & OS, incoh	\\
			\hline
		60 & 1.0 & 0.0001 & 0.5 &OS, incoh& 2 OD, coh  \\
			\hline	
	80 & 1.0 & 0.0001 & 0.5 &  OS, incoh&	OS, incoh	 \\
			\hline
		60 & 1.0 & 0.0003 & 1.5 & OD,  coh &  2 OD,  coh  \\
			\hline	
			60 & 0.0 & 0.0003 & 1.5 & unbounded & unbounded  \\
			\hline
			80 & 0.0 & 0.0003 & 1.5 & unbounded & unbounded\\
			\hline
		\end{tabular}
		\caption{Summary of collective behavior of edge and bulk units for different parameter values,  OS stands for oscillatory states, OD for OD-states, the number $2$ for two clusters,  further characterized by coherent (coh) or incoherent (incoh) behavior.}
		\label{tab:oscillator_behavior}
	\end{table}

\par The first three rows of the table show the situation with oscillations restricted to the edge, independent of $\alpha$ (first and third row) and for the same coupling ratios $w/s$ (first and second row, but $g$ is changed). 

The difference between the results of the third and fourth row indicates that by a single change of the linear self-coupling $g$ from $0.0$ to $1.0$ it is possible to switch from edge-localized oscillations and OD-states in the bulk (third row) to incoherent oscillations of the bulk and two groups of coherent OD-states along the edge (fourth row). Note that this result is obtained without a change of the coupling ratio $w/s$, by which we usually enforced edge-localized oscillations. This behaviour is illustrated in Fig.\ \eqref{Picture23}, for which we have chosen a $10 \times 10$ lattice of coupled activator-repressor units as given in Eq.\ \eqref{BFU}. 
Self-feedback is chosen as $g=1$ or $g = 0$. 
After an initial transient, the edge oscillators settle into OD-states, as shown in Fig.\ \eqref{Picture23}(a), whereas the bulk oscillators continue to oscillate, as shown in Fig.\ \eqref{Picture23}(b). At time $t = 2000$ we then use the state values as the initial conditions for the coupled system 
and turn off the self-feedback by setting $g = 0$. Clearly, in this case, the oscillators at the boundary of the 2D square lattice begin oscillating, while those in the interior transition to an OD-state.
\begin{figure}[htp]
	\centerline{\includegraphics[width=0.4\textwidth]{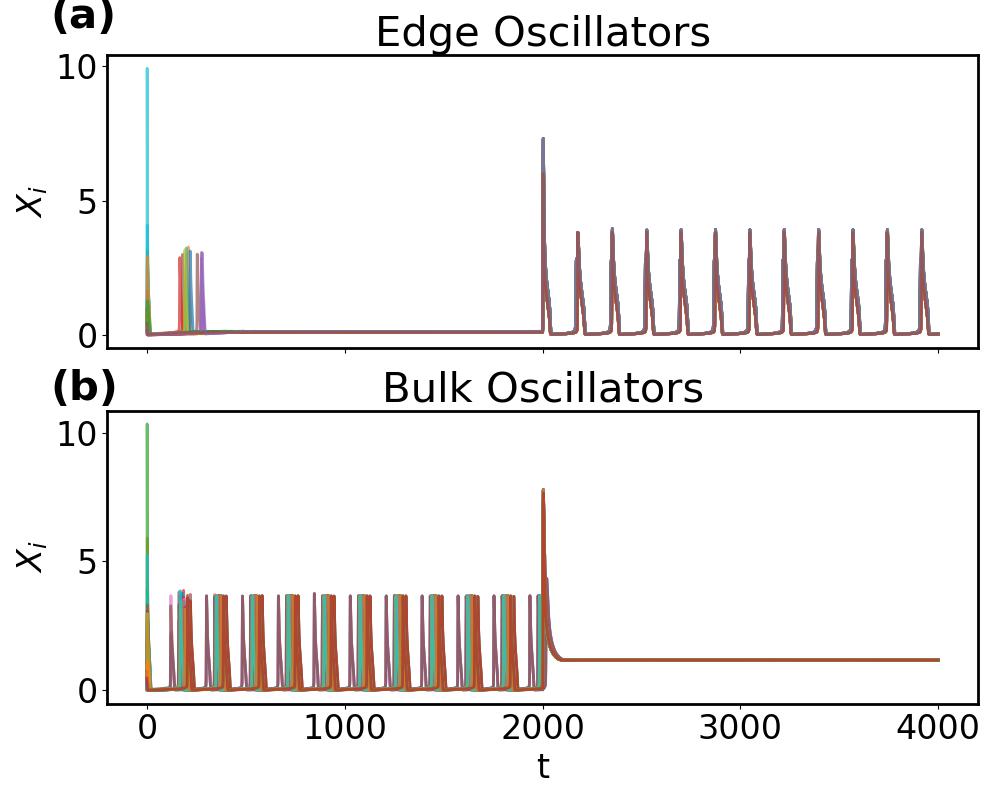}}
	\caption{{\bf Self-feedback induced switch between edge and bulk oscillations for coupled activator-repressor units}:  The initial conditions on a $10 \times 10$ lattice are chosen randomly from the range $[0,1] \times [0,1]$. Until time $t=2000$, we set $g=1$, with edge units at rest after an initial short transient of oscillations (a), while bulk units are oscillating (b);  after $t=2000$ we set  $g=0$  and find reversed states, that is, oscillations restricted to the edge (a) and bulk units in OD-states (b).   The other parameters are kept fixed at $s = 0.5$, $w = 0.0001$, $\alpha_i = 60$, $b = 0.01$, $\gamma = 0.01$, and $K = 0.02$.
	}
	\label{Picture23}
\end{figure}

Results for the $5$th and $6$th rows indicate states where all units are in oscillatory (OS) or OD-states, and where the ``$2$'' indicates two clusters of fixed-point values. For the parameter choice of the last two rows the states are unbounded.
 \\

{\bf Robustness of edge-oscillations.} To further validate the robustness of edge-oscillations, we introduce parameter mismatches and observe that the results remain consistent when $\alpha_i$ is uniformly distributed around the mean value $\alpha = 80.0$ with a broad distribution width of $\Delta \alpha = 10.0$. We present the dynamical evolution of each oscillator in Fig.\ \eqref{Picture27}, where we plot \( X_i \) vs \( Y_i \) in the phase space after a sufficiently large initial transient. In subfigure (a), the edge oscillators continue to oscillate and remain incoherent, exhibiting nonzero normalized frequencies. In contrast, subfigure (b) shows that the bulk oscillators transition to diverse OD-states, depending on $\alpha_i$. Since the OD-state in the activator-repressor model requires $X_i = Y_i$ (see Eq.\ \eqref{BFU}), the bulk oscillators settle along the line $X_i = Y_i$. Consequently, the system self-organizes into two distinct groups: incoherent edge oscillators and coherent bulk oscillators in a resting state, forming a  chimera-like state with respect to the normalized frequencies.

\begin{figure}[htp]
	\centerline{\includegraphics[width=0.5\textwidth]{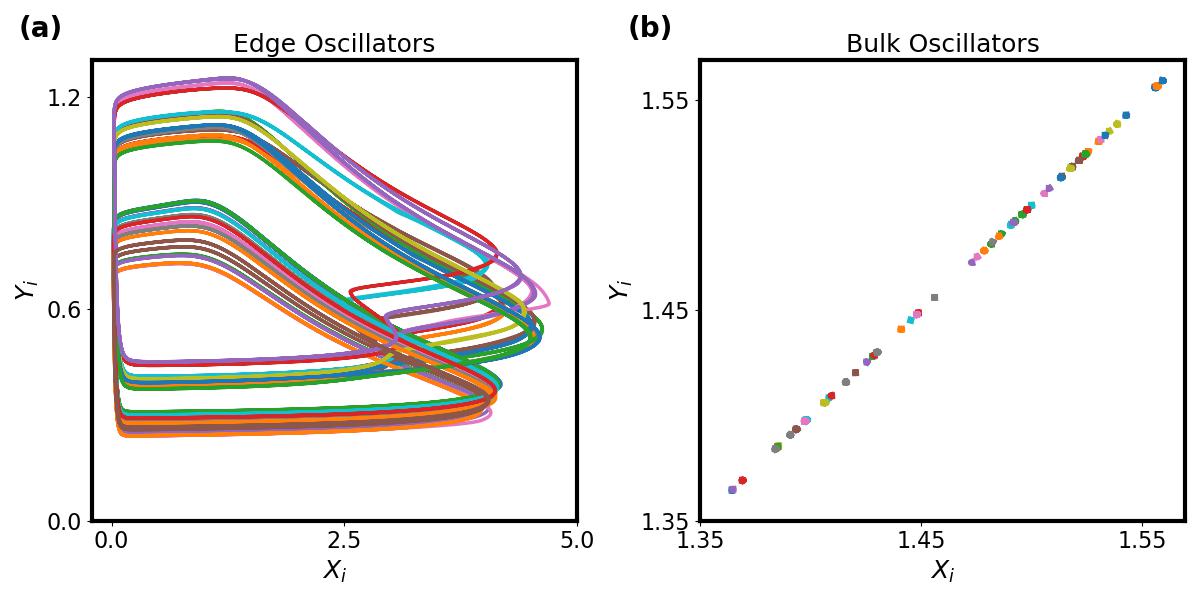}}
	\caption{{\bf Edge oscillations and bulk OD-states with parameter mismatch}: $10 \times 10$ lattice of coupled activator-repressor units with coupling assignment according to Fig.\ \eqref{Picture2}(a). The initial conditions are randomly selected from the range $[0,1] \times [0,1]$, parameter mismatch is introduced through $\alpha_i$, which follows a uniform distribution centered at $\alpha = 80.0$ with a distribution width of $\Delta \alpha = 10$. After an initial transient phase, (a) the edge oscillators sustain oscillations, (b) while the bulk units transition to OD-states.  The remaining parameters are fixed at $s = 0.5$, $w = 0.0001$, $b = 0.01$, $\gamma = 0.01$, and $K = 0.02$.
	}
	\label{Picture27}
\end{figure}

Furthermore, the qualitative nature of the results persists even in the presence of minor structural defects in the designed lattice, which we verify by  removing a few grid points from the 2D square lattice. We further introduce additive noise with a strength of up to \( 10^{-5} \), following the procedure outlined earlier for coupled SL-oscillators in Sec.\ \eqref{subsub342}. Our findings reveal that the system's qualitative behavior remains unchanged despite the presence of noise of suitable strength. 
In the Appendix \eqref{app3}. we derive the effective Hamiltonian, discuss its symmetries and determine the Zak phases, which reveal the topological protection of edge-oscillations and explain their resilience against perturbations as for SL-oscillators.

\subsection{Brusselator on-site dynamics}\label{subsec352}
\par In the brusselator model, the state variables represent concentrations of chemical species, which must remain non-negative as negative concentrations have no physical meaning. The condition $b_i > 1 + a_i^2$ ensures that an isolated brusselator is in an oscillatory state \cite{lefever1971chemical}.  We start with identical brusselators with parameter values $a_i = a = 1.0$ and $b_i = b = 3.0$, so that the condition $b > 1 + a^2$ is satisfied.  Additionally, we set $g_i = g = 0$, meaning that the linear self-coupling is turned off. The initial conditions for each oscillator are randomly selected from the range $[0,1] \times [0,1]$, and the coupling assignment is chosen according to Fig.\ \eqref{Picture2}(a). To analyze the collective behavior of coupled brusselators, we consider a $10 \times 10$ square lattice, where the primary control parameters are again the coupling strengths $s$ and $w$. 
\begin{figure}[htp]
	\centerline{\includegraphics[width=0.45\textwidth]{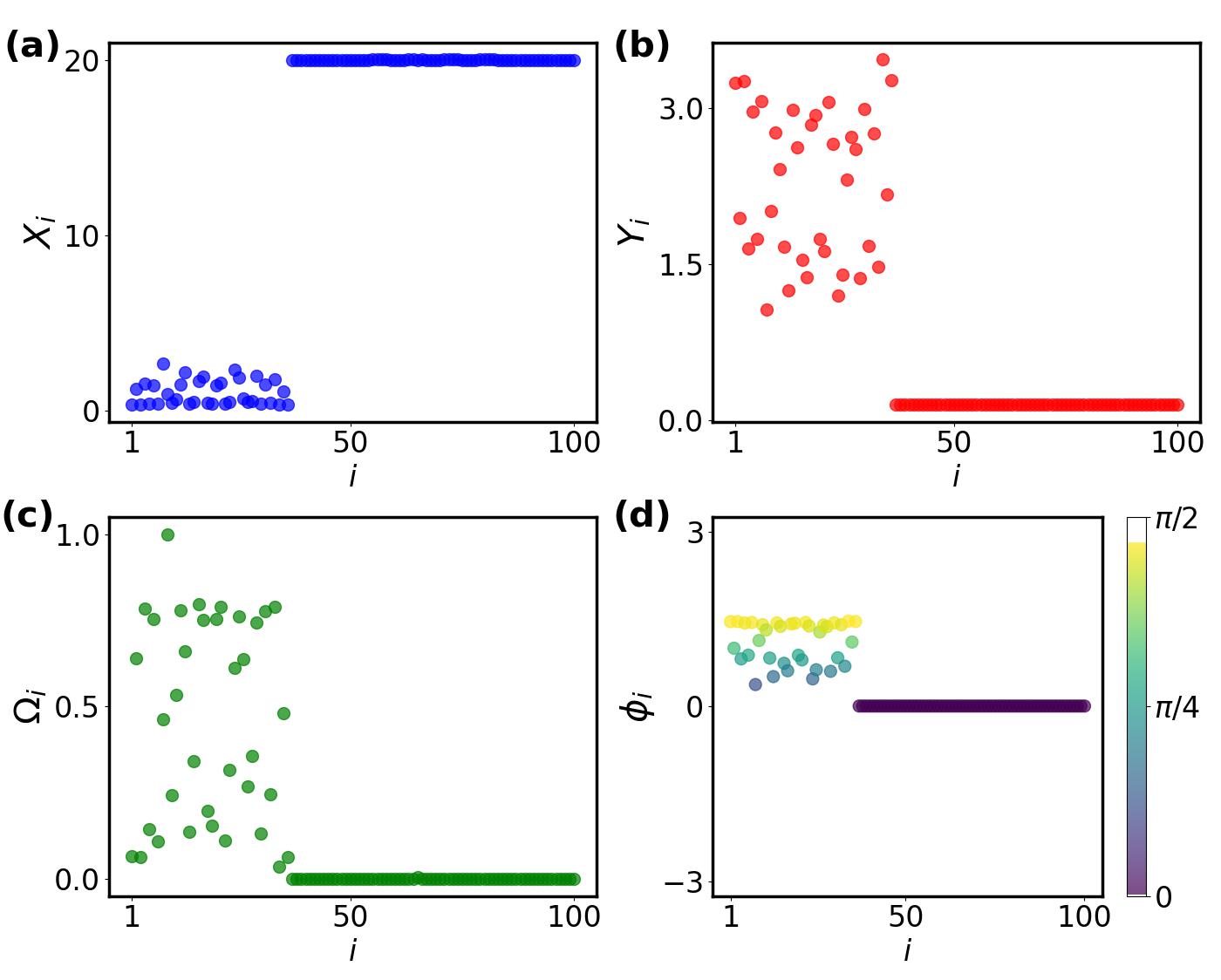}}
	\caption{{\bf Chimera-like state on a 2D square lattice of coupled brusselators}:   (a, b) Snapshots of the variables \( X_i \) and \( Y_i \), respectively, after the system has evolved beyond the transient regime.  
		(c) Normalized frequency \( \Omega_i \) showing sustained oscillations for edge oscillators, while bulk oscillators transition to an OD-state.  
		(d) Phase distribution \( \phi_i \), illustrating that bulk oscillators are phase-locked in a single coherent state, while edge oscillators remain incoherent. Parameters: $s=0.95$, $w=0.0001$, $g=0$, $a=1$, and $b=3$.
	}
	\label{Picture24}
\end{figure}
\par First we consider uniform coupling  for three different choices,  
$s = w = 0.95, 0.8, 0.6$, the system described in Eq.\ \eqref{eq:coupled_system} with coupled brusselators becomes then unbounded.  
However, for $s = w = 0.5$, all oscillators transition to distinct OD-states, depending on the initial conditions. On the other hand, for very weak coupling strengths, $s = w = 0.0001$, all coupled oscillators exhibit sustained oscillations. Edge-localized oscillations are observed when employing alternating weak and strong coupling strengths, specifically with $w = 0.0001$ and $s = 0.95$.  This behavior remains robust  when the strong coupling  
strength is slightly reduced to $s = 0.8$ while keeping $w = 0.0001$.  

We plot a snapshot of the oscillators' states, \( X_i \) and \( Y_i \), the normalized frequency values, \( \Omega_i \), and the phase of each oscillator, \( \phi_i \), at a specific time step after the system has evolved sufficiently beyond its initial transient in Fig.\ \eqref{Picture24}.  
Similar to previous snapshots, we first plot the indices of the edge oscillators, followed by the bulk oscillators along the $x$-axis. Clearly, we observe the expected edge oscillations, as indicated by the nonzero normalized frequency of the edge oscillators (subfigure (c)), while the bulk oscillators transition to an OD-state, where their normalized frequency becomes zero (subfigure (c)). In contrast to coupled SL-oscillators (Fig.\ \eqref{Picture4}), the bulk units converge into a single group (see (a) and (b)), thus, in brusselators we have a chimera-like state with respect to the state variables, nomalized frequencies, and phases (see Fig.\ \eqref{Picture24}).  This behavior remains qualitatively the same for \( s = 0.8 \) and \( w = 0.0001 \). 
\begin{figure}[htp]
	\centerline{\includegraphics[width=0.45\textwidth]{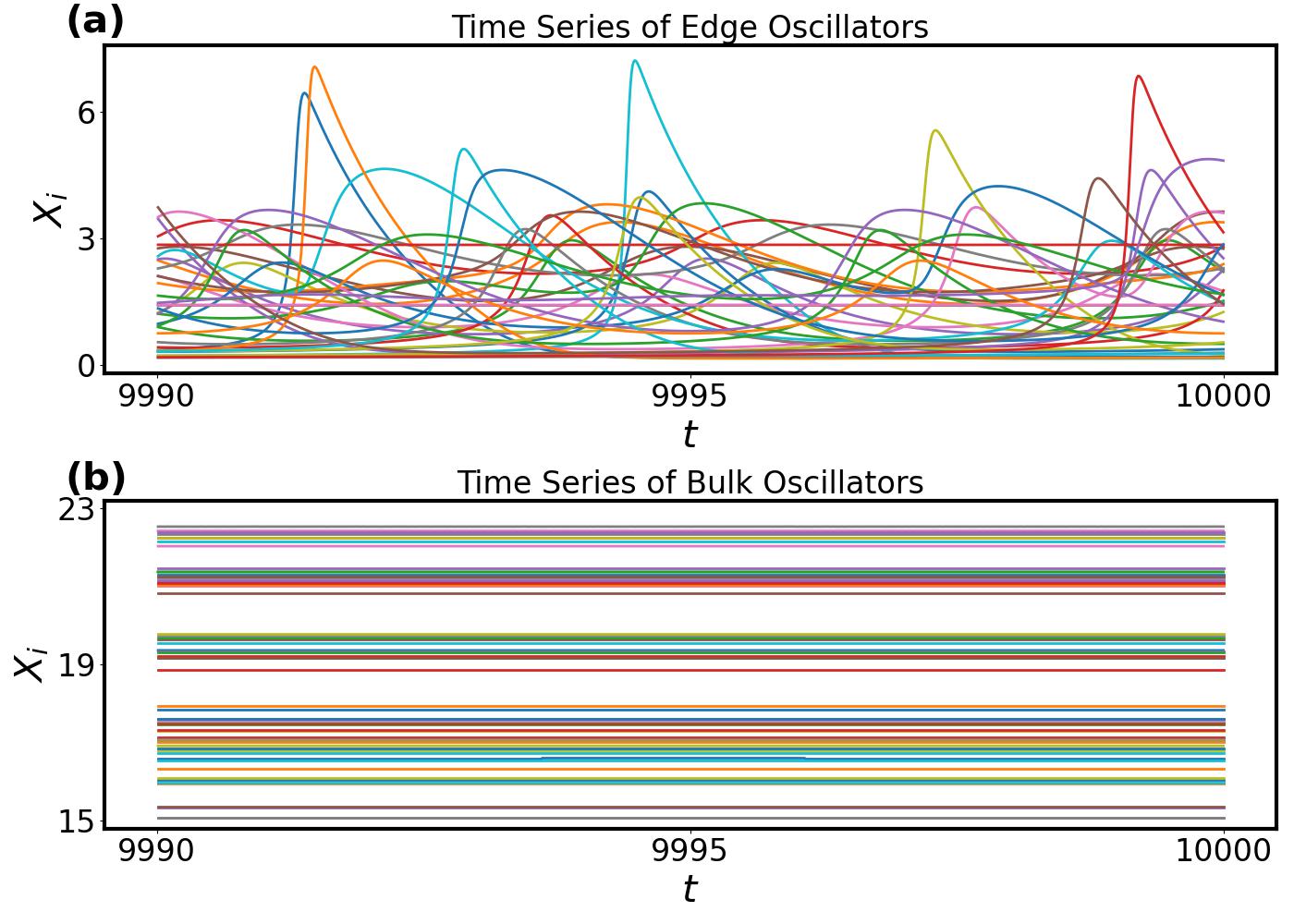}}
	\caption{{\bf  Time series for coupled brusselators with parameter mismatch:} Mismatch is introduced  in \( a_i \)  chosen from a uniform distribution centered at \( a = 1.0 \) with a distribution width of \( \Delta a = 0.5 \), while \( b \) remains fixed at \( b = 3.0 \). Some brusselators  along the edge have an \( a \) value for which \( b < 1 + a^2 \), causing those oscillators to automatically transition to an OD-state, while the majority of edge oscillators sustain oscillations. Other parameters are kept fixed at \( s = 0.95 \), \( w = 0.0001 \), and \( g_i = g = 0 \).
	}
	\label{Picture26}
\end{figure}

{\bf Robustness of edge-oscillations.} To further assess the robustness of our findings, we examine the impact of parameter mismatches. We find that our desired state remains robust when \( a_i \) is chosen from a uniform distribution centered at \( a=1.0 \) with a distribution width of \( \Delta a=0.4 \), while keeping \( b=3.0 \). Similarly, robustness is observed when \( b_i \) is drawn from a uniform distribution centered at \( b=3.0 \) with a distribution width of \( \Delta b=0.5 \), while keeping \( a=1.0 \). However, it is important to note that a careful selection of parameter mismatches is necessary to satisfy the constraint \( b > 1+a^2 \). If this condition is violated, the corresponding oscillators may spontaneously transition to OD-states, regardless of their position on the 2D square lattice, see  Fig.\ \eqref{Picture26}. 
Additionally, we observe robustness of edge oscillations when the lattice contains a few missing grid points or additive noise is introduced with a strength of up to \( 10^{-4} \), both figures are not displayed. 

\section{Summary and Conclusions}\label{sec6}

We have studied dynamical units that are individually either in a resting state or behave as classical nonlinear oscillators. Their coupling geometry is inspired by systems of condensed matter physics such as the SSH model, where dynamics is assigned to the links of a grid and edge phenomena like topologically protected edge currents can occur. From this perspective, it is of interest whether edge-localized activity can persist when nonlinear on-site dynamics is introduced in addition to the couplings along the connecting links. In our case, the edge-localized activity consists of oscillations that are, in general, not synchronized.

For the on-site dynamics, we have considered three types of oscillatory systems—Stuart--Landau units, activator--repressor units, and brusselators—all of which are individually known for their possible relevance in biological or biochemical applications. When the ratio of weak to strong couplings is appropriately chosen, edge-localized oscillations emerge robustly, regardless of the specific on-site dynamics, and persist under parameter mismatch, additive noise, or structural defects such as missing grid points. In most cases, we have edge-localized oscillations by tuning the ratio of inter- to intra-cell coupling. However, we have also observed a case in the activator--repressor system where additional self-coupling induces a switch between edge and bulk activity—a feature that may be further explored.

In all three cases, the derived effective Hamiltonian that  plays the role of a Bloch Hamiltonian exhibits inversion symmetry, resulting in quantized Zak phases of $(0,0)$ or $(\pi,\pi)$. These Zak phases match exactly with the observed behaviors: edge-localized oscillations for $(\pi,\pi)$ and OD-states throughout the grid for $(0,0)$, thus explaining the robustness of our numerical findings. Minor differences have appeared in the coherence of oscillations or OD-states. Generally, the edge oscillations are not synchronized except within smaller groups, and we have identified chimera-like states based on phase or normalized frequency.
We have also demonstrated how an appropriate removal of units from the grid can sustain oscillatory activity along closed bands inside a bulk of resting units, with sharp interfaces between active and resting units.

%

For future work promising extensions with respect to different parameter choices  are repressive couplings within or between unit cells. Particularly, for activator-repressor systems, repressive nonlinear coupling considerably enriches the dynamical behavior as we know from results for uniform coupling strengths \cite{labavic2014networks}. A similar enrichment of dynamics has also been reported for Stuart–Landau oscillators under the influence of repulsive coupling, where the interplay leads to cluster synchronization \cite{nag2024cluster} and  frustration \cite{chowdhury2020effect} in complex networks. For such a choice one may further search for signatures of chiral behavior along the edge related to wave propagation.\\
In view of biological applications, a first step may be to look for such applications only of unit cells, composed of four activator-repressor units.  In a next step one may then search for collective behavior that arises when many such units are  coupled to larger grids, whose edges might correspond to  membranes. The attractive feature would be that if one finds tunable  activity that can be localized to membranes, its  robustness would be pronounced  if topologically protected.

\section*{Acknowledgment} We thank Professor Evelyn Tang for valuable discussions in the beginning of this work and the German Research Foundation (DFG) (grant number ME-1332/30-1) for financial support.

\appendix
\section*{Appendix} \label{Appendix}

\appendix
\section{Edge oscillations for other explored parameter ranges}\label{AppendixA}


As shown in Fig.\ \eqref{Picture14}(b), the amplitude of bulk oscillators gradually increases with the coupling strength \( w \). For very weak coupling (\( w = 0.001 \) or \( 0.01 \)), and with all other parameters fixed as in Fig.\ \eqref{Picture14}, the dynamics resemble those in Figs.\ \eqref{Picture3} and \eqref{Picture4}. In this regime, edge oscillators exhibit frequency chimera-like behavior, while bulk oscillators tend toward oscillation death (OD), depending on initial conditions.

As \( w \) increases to around \( 0.1 \), bulk amplitudes exceed those at \( w = 0 \), and edge oscillators lose their frequency chimera-like coherence. Instead, they form an incoherent group with nonzero frequencies, while the bulk remains coherent with near-zero frequencies—still displaying a frequency chimera-like state, as in Fig.\ \eqref{Picture6}.

For higher values of \( w \) (e.g., \( w = 1 \) or \( 2 \)), bulk amplitudes continue to increase but remain lower than those of the edge oscillators (Fig.\ \eqref{Picture14}). Edge oscillators persist in incoherent oscillations, while the bulk remains coherent. Eventually, for sufficiently large \( w \), both edge and bulk oscillators may transition to OD states, depending on initial conditions.

Furthermore, we fix \( g = 0 \), \( \beta_i = 1.0 \), and \( \omega_i = \omega = 3.0 \) on a \( 10 \times 10 \) lattice arranged as in Fig.\ \eqref{Picture2}(a), and vary \( \alpha_i = \alpha \in [1, 6] \) with step size \( \delta\alpha = 1.0 \). We observe similar dynamics to those in Figs.\ \eqref{Picture3} and \eqref{Picture4}: edge oscillations persist while the bulk settles into OD.

Next, keeping \( \alpha_i = \alpha = 1.0 \) fixed, we vary \( \omega_i = \omega \in [1, 4] \) with \( \delta\omega = 1.0 \). Edge oscillators continue to display frequency chimera-like behavior, while bulk oscillators settle into two distinct OD states, depending on initial conditions (as in Fig.\ \eqref{Picture4}). For \( \omega = 5.0 \), bulk oscillators start oscillating within \( [-2, 2] \), while edge oscillators retain frequency chimera-like features.

Thus, for suitable ranges of \( \omega \) and \( \alpha \), edge oscillations can coexist with bulk OD. To investigate the role of the self-coupling term, we fix \( \omega_i = \omega = 3.0 \), \( \beta_i = 1.0 \), \( \alpha_i = \alpha = 1.0 \), \( w = 0.001 \), \( s = 8 \), and vary \( g_i = g \in [0, 3] \) with step \( \delta g = 1.0 \). As soon as \( g > 0 \), the edge oscillators become fully incoherent but remain oscillatory, while the bulk enters OD. This resembles a frequency chimera-like state, where incoherent edge oscillators (with nonzero frequencies) coexist with coherent bulk (at zero frequency). Increasing \( g \) further reduces the amplitude of edge oscillations:
for \( g = 0 \): range \( \sim [-2, 2] \) (Fig.\ \eqref{Picture3}),  \( g = 1 \): range \( \sim [-1, 1] \), 
	 \( g = 2 \): range \( \sim [-0.5, 0.5] \),  
	 \( g = 3 \): range \( \sim [-0.0005, 0.0005] \).
For \( g = 3 \), bulk oscillators begin oscillating incoherently, while edge oscillators exhibit multistability. Depending on initial conditions, edge oscillators either converge to amplitude death \cite{saxena2012amplitude} or oscillate with negligible amplitude near the origin.

\section{Symmetries of the Hamiltonian and their implications on the spectrum}\label{app2}
\label{appendixB}

\begin{lemma}\label{Theo1}
	Suppose, $\mathcal{H}(k_x, k_y)$ satisfies the inversion symmetry given in Eq.\ \eqref{inversion}. Then, the matrices $\mathcal{H}(k_x, k_y)$ and $\mathcal{H}(-k_x, -k_y)$ have the same eigenvalues.
\end{lemma}
\begin{proof}
	Let, $\mathcal{H}(k_x, k_y)\mathbf{v} = \lambda \mathbf{v}$; where $\lambda$ is an eigenvalue and $\mathbf{v}$ is the corresponding eigenvector.\\\\
	Now,  $\mathcal{I}(\mathcal{H}(k_x, k_y)\mathbf{v}) = \mathcal{I} (\lambda \mathbf{v}) $; where \(\mathcal{I} = \mathbf{\sigma}_x \otimes I_4\). So we have, 
	\begin{equation}\label{app1}
		 (\mathcal{I}\mathcal{H}(k_x, k_y))\mathbf{v} = \lambda (\mathcal{I} \mathbf{v}). 
	\end{equation}
By Eq.\ \eqref{inversion}, we have 	
	\begin{equation}\label{app2}
		\mathcal{I} \mathcal{H}(k_x, k_y) = \mathcal{H}(-k_x, -k_y)\mathcal{I}. 
	\end{equation}
	By combining Eqs.\ \eqref{app1} and \eqref{app2}, we get
		\begin{equation}\label{app3}
				 \mathcal{H}(-k_x, -k_y)(\mathcal{I}\mathbf{v}) = \lambda (\mathcal{I} \mathbf{v}). 
	\end{equation}
	Thus, the matrices \(\mathcal{H}(k_x, k_y)\) and \(\mathcal{H}(-k_x, -k_y)\) have the same eigenvalues. While the transformation preserves the eigenvalues, it alters the corresponding eigenvectors through the action of \(\mathcal{I}\). 
\end{proof}	

\begin{lemma}\label{Theo2}
	If $\mathcal{H}(k_x, k_y)$ satisfies the time-reversal symmetry given in Eq.\ \eqref{time-reversal}, then the eigenvalues of $\mathcal{H}(k_x, k_y)$ and $\mathcal{H}(-k_x, -k_y)$ come in complex conjugate pairs. That is, if $\lambda$ is an eigenvalue of $\mathcal{H}(k_x, k_y)$, then $\lambda^{*}$, the complex conjugate of $\lambda$, will be an eigenvalue of $\mathcal{H}(-k_x, -k_y)$.
\end{lemma}
\begin{proof}
		Let, $\mathcal{H}(k_x, k_y)\mathbf{v} = \lambda \mathbf{v}$; where $\mathbf{v}$ is an eigenvector of the matrix $\mathcal{H}(k_x, k_y)$ with an associated eigenvalue  $\lambda$.  
		Now, by Eq.\ \eqref{time-reversal}, we have 
		\begin{equation} \label{B4}  
			\mathcal{H}^{*}(k_x, k_y) = \mathcal{H}(-k_x, -k_y).  
		\end{equation}
		Now,  $\mathcal{H}(k_x, k_y)\mathbf{v} = \lambda \mathbf{v}$ $\implies$  $\mathcal{H}^{*}(k_x, k_y)\mathbf{v}^{*} = \lambda^{*} \mathbf{v}^{*}$. Using Eq.\ \eqref{B4}, we have
		\begin{equation} \label{app5}  
			\mathcal{H}(-k_x, -k_y)\mathbf{v}^{*} = \lambda^{*} \mathbf{v}^{*}.  
		\end{equation}
	which completes the proof.
\end{proof}


\begin{lemma}  \label{Theo4}
	If our effective Hamiltonian \(\mathcal{H}(k_x, k_y)\) satisfies both time-reversal symmetry, as defined in Eq.\ \eqref{time-reversal}, and inversion symmetry, as given in Eq.\ \eqref{inversion}, then it necessarily violates PT symmetry, which is expressed in Eq.\ \eqref{PT}.  
\end{lemma}  
\begin{proof}  
	Suppose, for the sake of contradiction, that the Hamiltonian \(\mathcal{H}(k_x, k_y)\) satisfies PT symmetry. Additionally, we assume that \(\mathcal{H}(k_x, k_y)\) also obeys time-reversal and inversion symmetry.
	From the PT-symmetry assumption, we have  
	\begin{equation}  
		(\mathcal{I}\mathcal{T}) \mathcal{H}(k_x, k_y) (\mathcal{I}\mathcal{T})^{-1} = \mathcal{H}(k_x, k_y).  
	\end{equation}  
	This implies
	\begin{equation}  
		(\mathcal{I}\mathcal{T}) \mathcal{H}(k_x, k_y) (\mathcal{T}^{-1}\mathcal{I}^{-1}) = \mathcal{H}(k_x, k_y).  
	\end{equation}  
	Applying the time-reversal symmetry condition, this simplifies to  
	\begin{equation}  
		\mathcal{I} \mathcal{H}^{*}(-k_x, -k_y) \mathcal{I}^{-1} = \mathcal{H}(k_x, k_y).  
	\end{equation}  
	Using inversion symmetry, we then obtain  
	\begin{equation}  
		\mathcal{H}^{*}(k_x, k_y) = \mathcal{H}(k_x, k_y).  
	\end{equation}  
	This implies that \(\mathcal{H}(k_x, k_y)\) contains no complex terms, which contradicts the fact that the Hamiltonian has complex components. Hence, our initial assumption that \(\mathcal{H}(k_x, k_y)\) satisfies PT symmetry must be false, and the result follows.  
\end{proof}

\section{Effective Hamiltonian and band structure for activator-repressor systems}\label{app3}
To compute the Zak phases, we first derive the effective Hamiltonian. Following the same procedure as for SL-oscillators, the effective Hamiltonian is given by  
	\small
	\begin{equation}\label{BFU_Hamiltonian}
		\mathcal{H}(k_x, k_y) = 
				\begin{bmatrix}
					\mathcal{H}_{11} & \mathcal{H}_{12} & w e^{ik_y}& 0 & 0 & 0 & s  & 0 \\ 
					
					\mathcal{H}_{21} & \mathcal{H}_{22} & 0 & 0 & 0 & 0 & 0 & 0 \\
					
					s  & 0 & \mathcal{H}_{33} & \mathcal{H}_{34} & w e^{ik_x}& 0 & 0 & 0 \\
					
					0 & 0 & \mathcal{H}_{43} & \mathcal{H}_{44} & 0 & 0 & 0 & 0 \\
					
					0 & 0 & s  & 0 & \mathcal{H}_{55} & \mathcal{H}_{56} & w e^{-ik_y}& 0 \\
					
					0 & 0 & 0 & 0 & \mathcal{H}_{65} & \mathcal{H}_{66} & 0 & 0 \\
					
					w e^{-ik_x}& 0 & 0 & 0 & s  & 0 & \mathcal{H}_{77} & \mathcal{H}_{78} \\
					
					0 & 0 & 0 & 0 & 0 & 0 & \mathcal{H}_{87} & \mathcal{H}_{88}
				\end{bmatrix},
	\end{equation}
where
\begin{eqnarray}
		\mathcal{H}_{11} &=& \frac{2 X_A \alpha (1 - b)}{(1 + (Y_A/K)) (1 + X_A^2)^2} - 1 - g, \\\nonumber
		\mathcal{H}_{12} &=& -\frac{\alpha}{K} \frac{b + X_A^2}{(1 + X_A^2) (1 + Y_A/K)^2}, \\ \nonumber
		\mathcal{H}_{21} &=& \gamma, \quad \mathcal{H}_{22} = -\gamma, \\ \nonumber
		\mathcal{H}_{33} &=& \frac{2 X_B \alpha (1 - b)}{(1 + (Y_B/K)) (1 + X_B^2)^2} - 1 - g, \\ \nonumber
		\mathcal{H}_{34} &=& -\frac{\alpha}{K} \frac{b + X_B^2}{(1 + X_B^2) (1 + Y_B/K)^2}, \\ \nonumber
		\mathcal{H}_{43} &=& \gamma, \quad \mathcal{H}_{44} = -\gamma, \\ \nonumber
		\mathcal{H}_{55} &=& \frac{2 X_C \alpha (1 - b)}{(1 + (Y_C/K)) (1 + X_C^2)^2} - 1 - g, \\ \nonumber
		\mathcal{H}_{56} &=& -\frac{\alpha}{K} \frac{b + X_C^2}{(1 + X_C^2) (1 + Y_C/K)^2}, \\ \nonumber
		\mathcal{H}_{65} &=& \gamma, \quad \mathcal{H}_{66} = -\gamma, \\ \nonumber
		\mathcal{H}_{77} &=& \frac{2 X_D \alpha (1 - b)}{(1 + (Y_D/K)) (1 + X_D^2)^2} - 1 - g, \\
		\mathcal{H}_{78} &=& -\frac{\alpha}{K} \frac{b + X_D^2}{(1 + X_D^2) (1 + Y_D/K)^2}, \\ \nonumber
		\mathcal{H}_{87} &=& \gamma, \quad \mathcal{H}_{88} = -\gamma
\end{eqnarray}
with $\alpha = \alpha_i$, and $g = g_i$ $\forall$ $i$.
The Hamiltonian $\mathcal{H}(k_x, k_y)$ satisfies time-reversal symmetry and obeys inversion symmetry under the condition $(X_A, Y_A) = (X_C, Y_C)$ and $(X_B, Y_B) = (X_D, Y_D)$, as discussed earlier. Fixing all parameters at $g = 0$, $s = 0.5$, $w = 0.0001$, $\alpha_i = 80$, $b_i = 0.01$, $\gamma_i = 0.01$, and $K_i = 0.02$, we numerically compute the equilibrium values:  $(X_A, Y_A) = (X_B, Y_B) = (X_C, Y_C) = (X_D, Y_D) = (1.4737,1.4737)$.

The band structure analysis reveals that for this choice of parameters, no exceptional points emerge.
The eigenvalue spectrum shows that bands $1,4,5,8$ are purely real, while bands $2-3$ and $6-7$ form complex conjugate pairs, not displayed here. A qualitatively similar feature has already been observed in the upper two rows of Fig.\ \eqref{fig102104} for coupled SL-oscillators. The computed Zak phases are $(\pi, \pi)$ for $w=0.5$, $s=0.0001$ for each of the eight bands, and $(0,0)$ for $w=0.0001$, $s=0.5$. This confirms that the observed state, where edge oscillators remain active while bulk oscillators transition to OD-states, is topologically protected. 
\begin{figure*}[htp]
	\centerline{\includegraphics[width=1.0\textwidth]{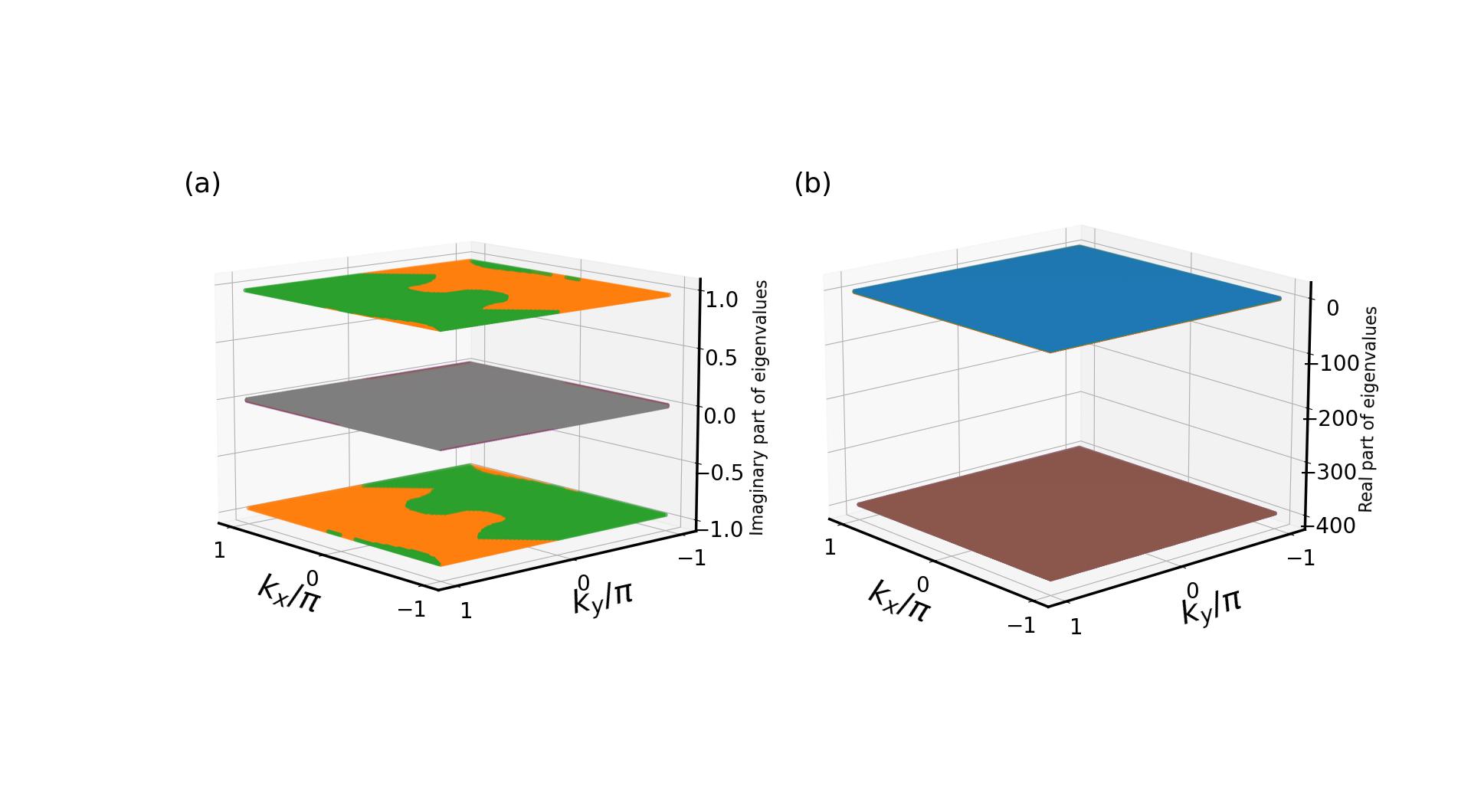}}
	\caption{{\bf Eigenvalue spectrum of the effective Hamiltonian \(\mathcal{H}(k_x, k_y)\) for coupled brusselators}: 
 (a) The imaginary parts of the eigenvalues split into two 
 parts (green and orange) due to eigenvalue switching. (b) The flat bands corresponding to the real part of the eigenvalues predominantly lie in the negative range, ensuring system stability. Despite an apparent visual 
coincidence, a closer inspection of (b) shows that bands 1–4 and  5–8 form tightly packed groups, with complex-conjugate pairings leading to six distinct real bands. 
Coupling parameters are \( w=0.0001 \) and \( s=0.95 \), other parameters are $g=0$, $a=1$, and $b=3$.
	}
	\label{Picture25}
\end{figure*}

\section{Effective Hamiltonian and band structure for brusselators}\label{app4}

The components of the effective Hamiltonian of Eq.\ \eqref{BFU_Hamiltonian}  for brusselator on-site dynamics are given as:
\begin{equation}
	\begin{aligned}
		\mathcal{H}_{11} &= 2X_A Y_A - b - 1 - g, \;\;
		\mathcal{H}_{12} = X_A^2, \\
		\mathcal{H}_{21} &= b - 2 X_A Y_A, \;\;
		\qquad \quad\;\mathcal{H}_{22} = -X_A^2, \\
		\mathcal{H}_{33} &= 2X_B Y_B - b - 1 - g, \;\;
		\mathcal{H}_{34} = X_B^2, \\
		\mathcal{H}_{43} &= b - 2 X_B Y_B, \;\;
		\qquad \quad\;\mathcal{H}_{44} = -X_B^2, \\
		\mathcal{H}_{55} &= 2X_C Y_C - b - 1 - g, \;\;
		\mathcal{H}_{56} = X_C^2, \\
		\mathcal{H}_{65} &= b - 2 X_C Y_C, \;\;
		\qquad \quad\;\mathcal{H}_{66} = -X_C^2, \\
		\mathcal{H}_{77} &= 2X_D Y_D - b - 1 - g, \;\;
		\mathcal{H}_{78} = X_D^2, \\
		\mathcal{H}_{87} &= b - 2 X_D Y_D, \;\;
		\qquad \quad\;\mathcal{H}_{88} = -X_D^2.
	\end{aligned}
\end{equation} We calculate the locally stable OD-states as
$(X_A, Y_A) = (X_B, Y_B) = (X_C, Y_C) = (X_D, Y_D) = (20.0, 0.15)$,
using Newton's method for the unidirectionally coupled four Brusselators with $s = 0.95$. This state changes to $ (X_A, Y_A) = (X_B, Y_B) = (X_C, Y_C) = (X_D, Y_D) = (5.0, 0.6)$,
for $s = 0.8$. Again, the effective Hamiltonian $\mathcal{H}(k_x, k_y)$ satisfies time-reversal symmetry and obeys inversion symmetry under the condition $(X_A, Y_A) = (X_C, Y_C)$ and $(X_B, Y_B) = (X_D, Y_D)$, which holds for these specific parameter values as well. In both cases, we compute the Zak phases \eqref{eqzak}, which remain $(0,0)$ for each band. 
When we swap the values and choose \( s = 0.0001 \) and \( w = 0.95 \) (or \( 0.80 \)), all oscillators transition to OD-states. In this case, the corresponding Zak phases become $(\pi, \pi)$ for each band. 

\par We further analyze the eigenvalue spectrum of the effective Hamiltonian \( \mathcal{H}(k_x, k_y) \) by varying \( k_x \in [-\pi, \pi] \) and \( k_y \in [-\pi, \pi] \), and plot the results in Fig.\ \eqref{Picture25} for the chosen coupling parameters \( w=0.0001 \) and \( s=0.95 \). We consider a total of \( k_{\text{points}} = 100 \) equally spaced points along each momentum direction. At first glance, Fig.\ \eqref{Picture25}  seems to suggest the presence of only two bands with nonzero imaginary parts in subfigure (a) and two distinct real bands in subfigure (b). This might lead one to anticipate possible eigenvalue degeneracies or band crossings for these parameter values. However, this apparent ambiguity arises primarily due to the large eigenvalue range \((0,-400]\) in subfigure (b), which masks the finer details of the band structure.

A closer inspection, particularly by zooming into subfigure (b), reveals that bands 1–4 are closely spaced in the range \([-2, 0)\), where bands 2 and 3 coincide as real part of the complex-conjugate pairs. Similarly, bands 5–8 cluster in the range \([-397.02, -396.98]\), with bands 5 and 8 also forming complex-conjugate pairs. Consequently, the system exhibits six real bands: four of them are purely real, while the remaining two are components of complex-conjugate quartets. This confirms the absence of exceptional points for \( w=0.0001 \) and \( s=0.95 \).

Apart from these spectral characteristics, the overall band structure qualitatively resembles that of the coupled SL-oscillators presented in Fig.\ \eqref{fig102104}. Notably, the bands associated with the nonzero imaginary eigenvalues in subfigure (a)
split into two unequal halves due to the eigenvalue switching. Meanwhile, the bands corresponding to the real eigenvalues in subfigure (b) remain predominantly negative, indicating system stability.

\section{Zak phases for a unit cell with weak couplings and corresponding lattice structures}\label{app5}

For completeness, we present here the Hamiltonian obtained when the unit cell is chosen with weak couplings assigned to its internal links.
Following the same derivation as described in Sec.\ \eqref{sec4}, the resulting Hamiltonian in momentum space takes the form
\begin{equation}\label{ouroldhamiltonian}
	\mathcal{H}(k_x, k_y) = 
	\scalebox{0.8}{$
		\begin{bmatrix}
			J_{XX}^{A} & J_{XY}^{A} & w  & 0 & 0 & 0 & se^{-i k_x} & 0 \\
			J_{YX}^{A} & J_{YY}^{A} & 0 & 0 & 0 & 0 & 0 & 0 \\
			s e^{i k_y} & 0 & J_{XX}^{B} & J_{XY}^{B} & w  & 0 & 0 & 0 \\
			0 & 0 & J_{YX}^{B} & J_{YY}^{B} & 0 & 0 & 0 & 0 \\
			0 & 0 & s e^{i k_x} & 0 & J_{XX}^{C} & J_{XY}^{C} & w  & 0 \\
			0 & 0 & 0 & 0 & J_{YX}^{C} & J_{YY}^{C} & 0 & 0 \\
			w  & 0 & 0 & 0 & s e^{-i k_y} & 0 & J_{XX}^{D} & J_{XY}^{D} \\
			0 & 0 & 0 & 0 & 0 & 0 & J_{YX}^{D} & J_{YY}^{D}
		\end{bmatrix}.
		$}
\end{equation}
As stated in the main text,  if the links, connecting nodes from the edge to nodes from the edge of the bulk, have weak couplings assigned, we observe edge-localized oscillations (here as for  Fig.~\ref{Picture2e}(b) and (d)), if these couplings  are strong as for Fig.~\ref{Picture2e}(a) and (c), we observe OD-states everywhere.

In more detail, let us keep $\alpha_i, \omega_i, \beta_i, g$ fixed and choose $s=8.0$ and $w=0.001$. To observe OD-states everywhere as suggested by the vanishing Zak phases for our first choice of Hamiltonian Eq.\ \eqref{ourhamiltonian},  it seems first at odds that instead of OD-states we observe edge oscillations for Fig.~\ref{Picture2e} (d). Therefore, let us extend the geometry of Fig.~\ref{Picture2e}(d) to a $6\times 6$ grid with the same alternating coupling assignments as in Fig.~\ref{Picture2e}(d). In this case, for the $6\times 6$ grid of Fig.~\ref{Picture2e}(a), indeed we would observe  OD-states everywhere, as the links orthogonal to the edge, connecting to the nodes in the boundary of the bulk, have strong couplings assigned. Switching here the values of $s$ and $w$ to $s=0.001$ and $w=8.0$, corresponding to Zak phases $(\pi,\pi)$, leads to edge-localized oscillations on a $6\times 6$ lattice of Fig.~\ref{Picture2e}(b).\\
If, instead of switching the coupling values on a grid of the same size, we change the size of the grid say again to $4\times 4$, the single plaquette in the bulk is weakly coupled to the edge (without any switch) for $s=8.0$ and $w=0.001$, going along with edge-localized oscillations (Fig.~\ref{Picture2e}(d)). Here, the corresponding non-vanishing Zak phases $(\pi,\pi)$ for $s=8.0$ and $w=0.001$ are obtained from the alternative Hamiltonian Eq.\ \eqref{ouroldhamiltonian}, derived for  weak couplings inside a unit cell. \\

\begin{figure}[htp]
	\centerline{\includegraphics[width=0.4\textwidth]{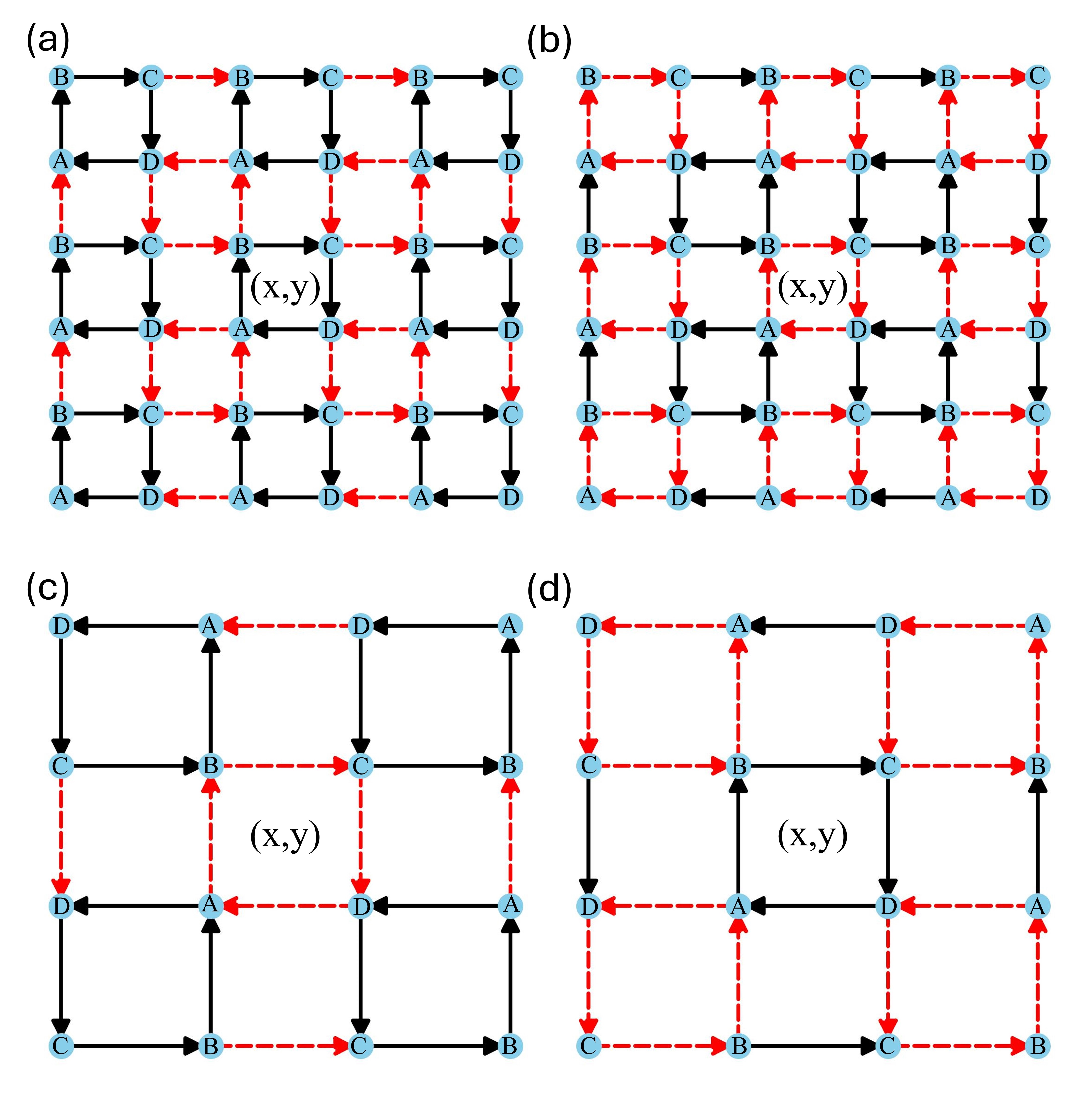}}
	\caption{{\bf Dependence on the lattice size}: Same as Fig.\ \eqref{Picture2}, but also for $6\times 6$ lattices. For further explanations see the text.
	}
	\label{Picture2e}
\end{figure}

This means, which Hamiltonian to choose (i.e., which coupling assignments to the internal links of a unit cell) for calculating the bulk bands and the appropriate Zak phases, depends on the distance between the unit cell and the edge. (The distance is measured as the minimal number of connecting links.) For Zak phases zero as in (a) and (c), the connecting links to the edge should be strong, therefore for an even (odd) distance of the edge from the unit cell ($2$ in (a) and $1$ in (c)), its internal links should be strong (weak), respectively; accordingly, for Zak phases $\pi$ as in (b) and (d), the connecting  links should be weak, therefore for an even (odd) distance from the unit cell ($2$ in (b) and $1$ in (d)), its internal links should be weak (strong), respectively. \\
For these choices, the non-vanishing Zak phases of the bulk bands of the Hamiltonian are in one-to-one correspondence to edge-localized oscillations on a finite grid with even or odd distances of unit cells to the edge. Thus, whether the distance is even or odd determines the choice of coupling assignment to the unit cell (and the  respective Hamiltonian).

\typeout{}
\bibliography{library}


\end{document}